\documentclass[12pt]{article}
\usepackage{natbib,graphicx,setspace}
\usepackage{mathrsfs,amsmath,amsthm,amssymb,color}
\usepackage{rotating}
\usepackage{booktabs}
\usepackage[ruled]{algorithm2e}
\usepackage{authblk}
\usepackage{tabularray}
\usepackage{multirow}
\usepackage{float}
\usepackage{ragged2e}
\usepackage{enumitem}
\usepackage{hyperref}
\hypersetup{
colorlinks=true,
linkcolor=blue,
citecolor=blue,
urlcolor=blue}

\usepackage{titlesec}
\titlespacing{\section}{0pt}{\parskip}{-\parskip}
\titlespacing{\subsection}{0pt}{\parskip}{-\parskip}
\titlespacing{\subsubsection}{0pt}{\parskip}{-\parskip}

\def\spacingset#1{\renewcommand{\baselinestretch}%
{#1}\small\normalsize} \spacingset{1}

\newtheorem{theorem}{Theorem}
\newtheorem{lemma}{Lemma}

\newcommand{\be}{\boldsymbol{\beta}}

\begin{document}

\title{\bf Distributed variable screening for generalized linear models}

\author[1,2]{Tianbo Diao}
\author[1]{Lianqiang Qu}
\author[1]{Bo Li}
\author[3,4]{Liuquan Sun}

\affil[1]{School of Mathematics and Statistics, Central China Normal University, Wuhan, Hubei, 430079, China}
\affil[2]{School of Mathematics and Science,
Nanyang Institute of Technology, Nanyang, Henan, 473004, China }
\affil[3]{Academy of Mathematics and Systems Science, Chinese Academy of Sciences, Beijing 100190, China}
\affil[4]{School of Mathematical Sciences, University of Chinese Academy of Sciences,  Beijing 100190, China}


\date{}
\maketitle

\begin{abstract}

	In this article,  we develop a distributed variable screening method for generalized linear models.
	This method is designed to handle situations where both the sample size and the number of covariates are large.
	Specifically, the proposed method selects relevant covariates by using a sparsity-restricted surrogate likelihood estimator.
    It takes into account the joint effects of the covariates rather than just the marginal effect, and this characteristic enhances the reliability of the screening results.
	We establish the sure screening property of the proposed method,
	which ensures that with a high probability, the true model is included in the selected model.
	Simulation studies are conducted to evaluate the finite sample performance of the proposed method,
	and an application to a real dataset showcases its practical utility.

\noindent {\bf KEY WORDS: } Distributed learning; Generalized linear models; Massive data; Variable screening.
\end{abstract}

\spacingset{1.5}
\section{Introduction}

With the rapid advancement of technology, we are frequently confronted with massive data in various scientific fields as well as in our daily lives, 
and these massive datasets often involve ultrahigh dimensional data.
The high dimensionality of such data presents simultaneous challenges in terms of computational cost, statistical accuracy, and algorithmic stability for classical statistical methods~\cite{fan2010selective}.
Therefore, feature screening is critically important.

Based on the analysis of marginal correlations between the response and covariates, ~\cite{fan2008sure} proposed the sure independence screening (SIS) method to select active covariates for ultrahigh dimensional linear models. 
The SIS method has been widely studied in various contexts 
due to its ability to rapidly reduce the dimension of covariates from a high dimensional space to a moderate scale 
that is smaller than the sample size~\citep{fan2010sure,zhu2011model,li2012feature}.

However, simple SIS methods face several challenges. For instance, they may overlook some crucial covariates that are individually unrelated 
but jointly associated with the response variable. Additionally, these methods can yield misleading outcomes when there are significant correlations among the covariates.
To overcome these issues, ~\cite{wang2009forward} introduced a forward regression screening method for linear regression models. 
Recently, ~\cite{zhou2020model} extended the idea of~\cite{wang2009forward} to model-free settings by incorporating the cumulative divergence, 
which effectively considered the joint impact of covariates. 
~\cite{xu2014sparse} proposed a jointly variable screening method by using the sparsity-restricted
maximum likelihood estimator, which naturally takes the joint effects of covariates in the screening process.
~\cite{yang2016feature} proposed a sure screening procedure for the Cox regression model 
using a joint partial likelihood function of potential active covariates.
~\cite{hao2021optimal} developed an optimal minimax variable screening method for large-scale matrix linear regression models.

The aforementioned methods are primarily designed under situations,
in which the number of covariates $p$ is large but the sample size $N$ is moderate.
When the number of samples and covariates are both huge, 
the classic screening methods may be computationally intractable 
due to storage constraints.  
To address these challenges, ~\cite{JMLR:v21:19-537} developed a distributed variable screening framework based on aggregated correlation measures.
More recently, ~\cite{li2023feature} proposed a distributed variable screening with conditional rank utility for big-data classification.
However, these methods considered the marginal effect of the covariates on the response 
and thus may fail to fully account for some crucial covariates.
We refer to~\cite{gao2022review} and~\cite{Zhou2024} for a comprehensive review on distributed statistical inference.

In this article, we develop a distributed variable screening method for the generalized linear models.
The proposed method shares the same spirit with~\cite{hao2021optimal}  and~\cite{xu2014sparse}, 
which considers the joint effects of the covariates rather than just the marginal effect.
Thus, it has a good potential to provide more reliable screening results.
The proposed method selects relevant covariates by using a sparsity-restricted surrogate likelihood estimator.
The surrogate likelihood is an approximate of the global likelihood
but is constructed using a local dataset.
Therefore, it significantly reduces the algorithmic complexity. 
For implementation, a distributed iterative hard-thresholding algorithm is developed specifically for the generalized linear models with large $p$ and large $N$.
We prove the convergence of the proposed algorithm
and also show the sure screening property of the proposed method.
This guarantees that the true model is contained in a set of candidate models selected by
the proposed procedure with overwhelming probability.
It should be pointed out that our method borrows ideas from the communication-efficient surrogate likelihood method studied by~\cite{jordan2019communicationefficient}.
However, this is not a straightforward task, 
as our method involves challenges in both computational methods and theoretical analysis, 
primarily because the surrogate likelihood considered here is constrained by an $\ell_0$ constraint.

The rest of this article is organized as follows. 
Section~\ref{method} introduces the distributed variable screening method. 
Section~\ref{SSP} establishes the theoretical properties of the proposed method. 
Simulation studies and real data analysis are provided in Section~\ref{SRA}. 
Concluding remarks are given in Section~\ref{CR}.
The technical details are included in the Appendix.

\section{Methodology}\label{method}
\subsection{Generalized Linear Models}

Let $ y \in \mathbb{R} $ be the response of interest and $ \boldsymbol{X} =(x_{1},\dots,x_{p})^\top \in \mathbb{R}^p $ denote the covariate vector. 
Throughout this article, we assume that the response $y$ is from an exponential family distribution with the density function
\begin{equation} \label{GLM model}
f (y; \theta) = \exp \left\{ \theta y - b(\theta)  + c(y) \right\}, 
\end{equation}
where $ b(\cdot) $ and $ c(\cdot) $ are known specific functions, and $ \theta$ represents the so-called natural parameter. 
We assume that $ \theta\in \bar{\Theta}$, 
where  $ \bar{\Theta} $ denotes a convex compact subspace of $\mathbb{R}^p$ 
and $ b(\cdot) $ is twice continuously differentiable.
Under model~\eqref{GLM model}, the mean of $ y $ is $ b^{\prime}(\theta)$, the first derivative of $ b(\theta) $ with respect to $ \theta $.
Recall that $\theta$ links the relation between $y$ and $\boldsymbol{X} $.
We consider the problem of estimating a $p$-dimensional vector of parameter $ \be = (\beta_1, \ldots, \beta_p)^{\top} $ from the following generalized linear model:
\[
\mathbb{E}(y \mid \boldsymbol{X} )=b^{\prime}(\theta(\boldsymbol{X} ))=g^{-1}\Big(\sum_{j=1}^p x_{j}\beta_j\Big),
\]
where $g(x)$ is a prespecified link function.
For simplicity of presentation, we consider the canonical link, that is, $g(\theta)=(b^{\prime}(\theta))^{-1}$.
The use of the canonical link function results in a simplified form with $ \theta(\boldsymbol{X} )= \boldsymbol{X}^{\top} \be $.

Assume that we observe $ N $ independent and identically distributed (i.i.d.) copies of $(\boldsymbol{X} ,y) $, 
denoted by $(\boldsymbol{X}_i,y_i), i \in \{ 1, \ldots, N \}$.
For model~\eqref{GLM model}, the negative log-likelihood is
given by
\[  \mathcal{L}(\be) =  \sum_{i=1}^{N} \left[  b(\boldsymbol{X}_{i}^{\top} \be ) - (\boldsymbol{X}_{i}^{\top} \be )y_{i} \right].  \] 

When $N$ is moderate and $p\ll N$, we can apply the maximum likelihood estimator (MLE) to estimate $ \be $, that is,
\begin{equation} \label{GLM loss}
\widehat{\be}=\text{arg}\min_{\be\in\mathbb{R}^p}~\mathcal{L}(\be). 
\end{equation}
The asymptotic properties of the MLE defined in \eqref{GLM loss} have been extensively researched in literature. 
However, minimizing \eqref{GLM loss} becomes challenging when both $p$ and $N$ are large. 
In the subsequent sections, we address these challenges by employing a distributed variable screening method specifically designed for scenarios with large $p$ and large $N$.

\subsection{Distributed variable screening method}
In this section, we develop a ``divide and conquer'' approach to select non-zero coefficients under model~\eqref{GLM model}. 
Let $S^* = \{j : \beta_j^* \neq 0 \}$ and $q=\text{card}(S^*)$, where $ \be^*=(\beta_1^*,\dots,\beta_p^*)^\top$ is the true value of $\be$
and $ \text{card}(S) $ denotes the cardinality of any set $S$. 
Assume that $q$ is significantly smaller than the sample size $N$, 
and the data is distributed and stored across $m$ machines. 
Without loss of generality, suppose that each machine contains $n = N/m$ observations.
Let $ \mathcal{Y}  = ( \boldsymbol {Y}_1^{\top}, \ldots ,\boldsymbol {Y}_m^{\top} )^{\top}$ and $\mathcal{X} = ( \mathbb{X}_1^{\top}, \ldots , \mathbb{X}_m^{\top} )^{\top} $
denote the full samples,
where $ \boldsymbol{Y}_i=(y_{i1}, \ldots ,y_{in})^{\top} $ and $ \mathbb{X}_i  = (\boldsymbol{X}_{i1}, \ldots , \boldsymbol{X}_{in})^{\top}$.
Also let $ \mathcal{T}_i = (\boldsymbol{X}_{ij}, y_{ij}),  j \in \{ 1, \ldots , n \} $ denote the subsample stored in the $ i $th machine. 
Then we can rewrite $\mathcal{L}(\be)$ as
\[
	\mathcal{L}(\be) = \frac{1}{N} \sum_{i=1}^{m} \sum_{j=1}^{n}\left[  b(\boldsymbol{X}_{ij}^{\top} \be ) - (\boldsymbol{X}_{ij}^{\top} \be )y_{ij} \right]
	= \frac{1}{m} \sum\limits _{i=1}^{m}\mathcal{L}_{i}(\be),
\] 
where $\mathcal{L}_{i}(\be) = n^{-1}\sum_{j=1}^{n}[b(\boldsymbol{X}_{ij}^{\top} \be )-(\boldsymbol{X}_{ij}^{\top} \be )y_{ij}],~  i \in \{ 1,\ldots,m \} $.

Inspired by~\cite{jordan2019communicationefficient}, we adopt a surrogate loss function to approximate the original loss function $\mathcal{L}(\be)$. 
Specifically, we approximate $\mathcal{L}(\be)$ using the following surrogate loss function formulation:
\[
\ell(\be)=\mathcal{L}_{1}(\be) - \langle \be, \nabla\mathcal{L}_{1}(\widetilde{ \be }) -\nabla\mathcal{L}(\widetilde{ \be }) \rangle,
\]
where  $ \langle a, b \rangle=a^\top b$ denote the inner product of vectors $a$ and $b$. 

For any vector $a=(a_1,\dots,a_p)^\top$, 
let $ \| a \|_0=\sum_{j=1}^p \mathbb{I}(a_j\neq 0)$ and $ \|a\|_2=(\sum_{j=1}^p a_j^2)^{1/2}$, 
where $ \mathbb {I}(\cdot) $ is the indicator function. 
Then under the sparsity assumption of $ \be^* $, we propose a variable screening method utilizing a sparse-restricted estimator defined by 
\begin{align} \label{initial problem}
\widehat{\be}_{[k]}= \underset{\| \be \|_0 \leq k}{\arg\min}\ \mathcal{L}_{1}(\be) - \langle \be, \nabla\mathcal{L}_{1}(\widetilde{ \be }) -\nabla\mathcal{L}(\widetilde{ \be }) \rangle,
\end{align}
where $\widetilde{ \be }$ is an estimator of $ \be $ and $ k $ is a prespecified constant.
By imposing the constraint $\| \be \|_0 \leq k$, our proposed method effectively selects a subset of variables with at most $k$ non-zero coefficients.
Let $ \widehat{S}_{[k]} $ be the selected model obtained by $\widehat{\be}_{[k]}$, that is, $ \widehat{S}_{[k]} = \{1\le j\le p: \widehat \beta_{[k]j}\neq 0\} $,
where $\widehat\beta_{[k]j}$ is the $j$th element of $\widehat\be_{[k]}$.
In what follows, we choose $\widetilde \be$ as the Lasso estimator based on $\mathcal{L}_{1}(\be)$. 
Specifically, we obtain $\widetilde \be$ by
\begin{align}\label{initial:Lasso}
    \widetilde{ \be } = \underset{\be\in\mathbb{R}^p}{\arg\min} \, \mathcal{L}_{1}(\be) + \lambda \left \Vert \be \right \Vert_{1},
\end{align}
where $\lambda >0$ is a tuning parameter.

Since the $\ell_0$-constrained problem has been proven to be NP-hard~\citep{natarajan1995sparse,chen2014complexity}, 
it is a considerable challenge to directly solve~\eqref{initial problem}.
To address this challenge, we adopt an iterative hard-thresholding (IHT) method to obtain $\widehat{\be}_{[k]}$. 
We begin by considering a quadratic approximation to the objective function $\ell(\be)$ as follows:
\[  
h(\boldsymbol{\gamma}; \be)= \ell(\be) + \langle \boldsymbol{\gamma} - \be, \nabla \ell(\be) \rangle + \frac{\vartheta }{2} \Vert \boldsymbol{\gamma} - \be \Vert_2^2,  
\] 
where $\vartheta >0$ represents the tuning parameter. It is worth noting that $ h(\be; \be)=\ell(\be) $, 
and for values of $\boldsymbol{\gamma}$ close to $\be$, $h(\boldsymbol{\gamma}; \be)$ provides a good approximation to $\ell(\be)$. 
Utilizing this quadratic approximation, we derive an iterative solution to problem~\eqref{initial problem}.
Let $\widehat{ \be }^{(t)}$ be an estimator of $\be^*$ at the $t$th iteration.
We update $\widehat{ \be }^{(t)}$ by $\widehat{ \be }^{(t+1)}$, where
\[  \widehat{\be}^{(t+1)} = \underset{\Vert \boldsymbol{\gamma} \Vert_0 \leq k}{\arg\min}~ h(\boldsymbol{\gamma}; \widehat{ \be }^{(t)}).  \]
This is equivalent to
\begin{equation}  \label{true iterative problem}        
\widehat{ \be }^{(t+1)}=\underset{ \Vert \boldsymbol{\gamma} \Vert_0 \leq k}{\arg\min}~ \big\Vert \boldsymbol{\gamma} - \big( \widehat{ \be }^{(t)}- \vartheta^{-1} \nabla \ell  (\widehat{ \be }^{(t)}) \big)   \big\Vert_2^2.  
\end{equation} 
Let $ \widehat{ \boldsymbol{\gamma} }^{(t)}= \widehat{ \be }^{(t)} -\vartheta^{-1} \nabla \ell (\widehat{ \be }^{(t)})$, and $\widehat\gamma_j^{(t)}$ denote its $j$th component.
By~\eqref{true iterative problem}, the $j$th component $\widehat\beta_j^{(t+1)}$ of $\widehat\be^{(t+1)}$ takes an explicit form
\begin{align} \label{eq5}
    \widehat\beta_j^{(t+1)}=\widehat\gamma_j^{(t)}\mathbb{I}(|\widehat\gamma_j^{(t)}|>r_k),
\end{align}   
where $r_k$ is the $k$th largest component of $|\widehat{\boldsymbol{\gamma} }^{(t)}|$.

The iterative procedure is summarized in Algorithm~\ref{DIHT algorithm}.
The proposed method involves a single round of data communication, 
which makes it highly efficient and easily implementable.

\begin{algorithm}[H]
	\caption{Distributed Iterative Hard-thresholding Algorithm (DIHT)}
	\label{DIHT algorithm}
	
	\KwIn{$\widetilde{ \be }$, $k$, an initial value $ \widehat{ \be }^{(0)}$ and a tolerance parameter $\varepsilon$.}	
		
	\textbf{Step 1:} Transfer $\widetilde{ \be }$ to each local machine.

	\textbf{Step 2:} For each machine $i \in \{ 1, \ldots, m \}$, calculate $ \nabla\mathcal{L}_{i}(\widetilde{ \be })$.

    \textbf{Step 3:} Transfer $ \nabla\mathcal{L}_{i}(\widetilde{ \be })$ to the first machine and compute 
        \[  \nabla\mathcal{L}(\widetilde{ \be })=m^{-1}\sum_{i=1}^m \nabla\mathcal{L}_{i}(\widetilde{ \be }).  \] 

    \textbf{Step 4:} Computes $ \widehat{ \be }^{(t+1)} $ by~\eqref{eq5}.
	
	\textbf{Step 5:} Stop the algorithm if $ \Vert \widehat{ \be }^{(t+1)}-\widehat{ \be }^{(t)} \Vert_2 \leq \varepsilon$. 
    Otherwise, let $t\leftarrow t+1$ and return {\bf Step 4}. 

	\KwOut{$\widehat\be_{[k]}=\widehat\be^{(t+1)}$.}
\end{algorithm}

\begin{theorem} \label{thm: IHT}
If $\vartheta \geq \rho_1 \mu/n$, then
$\ell (\widehat{ \be }^{(t+1)}) \leq \ell(\widehat{ \be }^{(t)}),$
where $ \rho_1 $ denotes the largest eigenvalue of $ \mathbb{X}_1^{\top}\mathbb{X}_1 $ and $ \mu = \max_{\theta \in \bar{\Theta}} b^{\prime \prime} (\theta) $.
\end{theorem}

Theorem~\ref{thm: IHT} shows that by appropriately selecting $\vartheta$,
the proposed procedure improves the current estimate within the feasible region ($\|\be\|_0\le k$).
In addition, Theorem~\ref{thm: IHT} indicates that $\vartheta$ should be large enough to ensure a decrease in $\ell(\be)$ after each iteration.
When $Y_i$ is sampled from a normal distribution, the second derivative of the canonical link function $b^{\prime \prime}(\theta)$ is equal to 1. Therefore, we can set $\vartheta=\rho_1$. 
In the case of logistic regression, a reasonable choice for $\vartheta$ is $\rho_1/4$, as the second derivative of the canonical link function is equal to 1/4.
Inspired by~\cite{xu2014sparse}, we apply an adaptive procedure to select $\vartheta$.
Specifically, we first choose a small initial value for $\vartheta$, 
and then check whether $\ell(\widehat{ \be }^{(t+1)}) \leq \ell(\widehat{ \be }^{(t)})$ or not. 
If this is violated, we double the current $\vartheta$ until $\ell(\widehat{ \be }^{(t+1)}) \leq \ell(\widehat{ \be }^{(t)})$ is met. 
It is trivial that the adaptive strategy does not alter the convergence of the proposed procedure.

The following theorem establishes the convergence property of the proposed algorithm.

\begin{theorem} \label{thm: convergence}
Let $\widehat{ \be }^{(t)}$ be the sequence generated by Algorithm~\ref{DIHT algorithm}. 
If the conditions stated in Theorem~\ref{thm: IHT} are satisfied, 
then there exists at least one subsequence of $\widehat{ \be }^{(t)}$ that is convergent as $ t \rightarrow \infty$. 
In addition, Algorithm~\ref{DIHT algorithm} stops with a finite number of steps.
\end{theorem}

\section{Sure screening property}\label{SSP}
In this section, we establish the sure screening property of the proposed method. 
All of the proofs are provided in the Appendix.
To do that, we fix some notation. 
For any subset $S$ of $ \{1, \ldots , p\} $, 
let $ \boldsymbol{X}_S =(x_j , j \in S)^{\top} $ and $ \be_S =(\beta_j , j \in S)^\top$. 
For any two positive sequences $ a_n $ and $ b_n $, 
we write $ a_n \gtrsim b_n $ and $ a_n \lesssim b_n $ if there exists a constant $ C > 0 $ independent of $ n $ 
such that $ a_n \ge C b_n $  and $ a_n \le C b_n $ for all sufficiently large $ n $, respectively. 
Moreover, we write $ a_n=o(b_n) $ if $ a_n/b_n\rightarrow 0$ as $n\rightarrow \infty$,
and $ a_n=O(b_n) $ if there exists a universal constant such that $a_n\le Cb_n$ for all $n$.

We make the following conditions.
\begin{enumerate}[label=(C\arabic*), series=(C)]
\item $ \ln p = O(N^\alpha) $ for some $ 0 < \alpha < 1 $.

\item There exist some positive constants $w_1, w_2, \tau_1$ and $\tau_2$ such that
\[
	\min_{j \in S^{*}} \left| \beta_{j}^{*} \right| \geq w_{1} N^{-\tau_{1}} \quad \text { and } \quad q<k \leq w_{2} N^{\tau_{2}}.
\] 

\item There exist some positive constants $ c_1 $ and $ \delta_1$ such that for sufficiently large $ n $,
\[  \phi_{\min} \left[ \mathcal{H} (\be_{S}) \right] \geq c_{1}  \]  
for any $ \be_S \in \left\lbrace  \be_S : \Vert \be_S - \be_S^* \Vert_2 \leq \delta_1 \right\rbrace  $ and $S\in \mathcal{M}_{+}^{2k} $, 
where $\mathcal{M}_{+}^{a}=\left\{S: S^{*} \subset S , \Vert S \Vert_{0} \leq a\right\}$ denotes the collection of the over-fitted models with size less than $a$,
$ \phi_{\min}(A)$ is the smallest eigenvalue of any matrix $A$, 
and 
\[  \mathcal{H} (\be_{S}) = \frac{\partial^2 \ell (\be_{S})}{\partial \be_{S} \partial \be_{S}^{\top}} 
= \frac{1}{n} \sum_{j=1}^{n} \boldsymbol{X}_{1jS} \, b^{\prime \prime}\left(\boldsymbol{X}_{1jS}^{\top} \be_{S} \right)\boldsymbol{X}_{1jS}^{\top}  \] 
is the Hessian matrix of $ \ell (\cdot) $ corresponding to $ \be_{S} $.

\item  There exist some positive constants $ c_2$  and $ c_3 $ such that $ \vert x_{ij,l} \vert \leq c_2 $
and  
\[  \max _{1 \leq i \leq m} \max_{1 \leq j \leq n}  \max_{1 \leq l \leq p}  
\left\{ \frac{x_{i j, l}^{2}}{\sum_{i=1}^{m} \sum_{j=1}^{n} x_{ij,l}^{2} \sigma_{ij}^{2}} \right\} \leq c_{3}  N^{-1}  \] 
for sufficiently large $ N $,
where $ \sigma_{ij}^2 = b^{\prime \prime}(\boldsymbol{X}_{ij}^{\top} \be^*) $
and  $x_{ij,  l} $ is the $ l $th component of $ \boldsymbol{X}_{ij} $.
\end{enumerate}

\begin{enumerate}
	\item[(C5)]  There exists a positive constant $c_4$  such that for sufficiently large $ n $,
		\[  \boldsymbol{w}^\top \mathcal{H}(\be) \boldsymbol{w} \geq  c_4 \Vert  \boldsymbol{w}_{S^{*}} \Vert_{2}^{2}  \]
		for any $ \boldsymbol{w} \neq 0 $ satisfying $ \Vert \boldsymbol{w}_{S^{*}_c} \Vert_{1} \le  3 \Vert \boldsymbol{w}_{S^{*}} \Vert_{1} $ and for any $ \be $
		such that $ \theta = \boldsymbol{X}^\top \be \in \bar{\Theta} $,
		where $ \mathcal{H}(\be) $ is the Hessian matrix of $ \ell(\be) $ and $ S^{*}_c=\{1\le j\le p: j\not\in  S^{*}\}$.
\end{enumerate}

Condition (C1) assumes that the number of covariates is allowed to grow exponentially with respect to the sample size $N$.
Condition (C2) imposes a requirement on the magnitude of non-zero coefficients, 
which is necessary to achieve the sure screening property. 
In simple terms, if certain non-zero coefficients decrease too rapidly and approach zero, 
it becomes challenging to distinguish them from the noise. 
The second part of (C2) limits the level of sparsity in the true coefficient vector $\beta^*$, 
but still permits the size of the true model to increase with the sample size $N$.
Condition (C3) is the uniform uncertainty principle (UUP) condition given by~\cite{candes2007dantzig}, 
which is mild and has been widely used in the literature of high-dimensional methods 
(see, e.g., \citep{chen2012extended,xu2014sparse}).
Condition (C3), together with (C2), ensures the identifiability of $ S^* $ over $ S$ for $\text{card}(S)\le k$
with a proper choice of $k$.
Condition (C4) is also employed by~\cite{xu2014sparse},
which is adopted to simplify the proofs of the following theorems.
Condition (C5) is the restricted eigenvalue condition~\citep{bickel2009simultaneous,xu2014sparse}, 
which ensures a desirable bound on the Lasso estimator. 
We refer to~\cite{bickel2009simultaneous} for more discussions and other similar conditions.

\begin{theorem} \label{thm: screening property}	
Suppose that  Conditions (C1)-(C5) hold, and $ \widetilde{ \be } $ is the Lasso estimator defined in~\eqref{initial:Lasso} with $ \lambda = O(q \sqrt{\ln p/n}) $. 
	If $ \tau_1 + \tau_2 < (1 - \alpha)/2 $ and $ n \gtrsim N^{(2/3) \tau_1 + 2\tau_2 + \alpha} $, 
	then we have
	\[  \Pr(S^* \subset \widehat{S}_{[k]})  \rightarrow 1 ~ \text{as}~ N \rightarrow \infty.  \] 
\end{theorem}

The proof of Theorem~\ref{thm: screening property} is given in the Appendix.
Theorem~\ref{thm: screening property} states that the proposed method enjoys the sure screening property,
that is, with probability tending to 1, all of the important variables are retained in the selected model $ \widehat{S}_{[k]} $.
Theorem~\ref{thm: screening property} requires the condition $ n \gtrsim N^{(2/3) \tau_1 + 2\tau_2 + \alpha} $,
which implies that the number of node machines $m$ satisfies $ m \lesssim N^{1-((2/3) \tau_1 + 2\tau_2 + \alpha)} $.
If $ \tau_1=0 $ and $ \tau_2=0 $, then $m$ can be of order $ N^{1-\alpha}$ ($0< \alpha<1$).
Thus, the proposed method still works with a large number of node machines.
However, we also observe that as the number of covariates $p$ increases, the number of node machines decreases. 
This is primarily due to the fact that the proposed method necessitates a reliable initial value $\widetilde\be$ 
in order to ensure the sure screening property.

Next, we establish the sure screening property of Algorithm~\ref{DIHT algorithm}. 
Let $ \widehat S^{(t)} = \{1\le j\le p: \widehat{ \beta }_{j}^{(t)} \neq 0 \} $. 
A good initial setup for $\widehat\be^{(0)}$ increases the chances of Algorithm~\ref{DIHT algorithm} 
reaching a local minimizer with the desired sure screening property. 
To achieve this, we select the Lasso estimator $\widetilde{\be}$ defined in~\eqref{initial:Lasso} 
as the initial value $\widehat\be^{(0)}$.

\begin{theorem} \label{thm: DIHT screening property}
Suppose that Conditions (C1), (C2), (C4) and (C5) hold, 
and $ \widetilde{ \be } $ is the Lasso estimator defined in~\eqref{initial:Lasso} with 
$ \lambda \sqrt{n/N^\alpha}  \to \infty $ and $ \lambda N^{\tau_1 + \tau_2} \to 0$.
If $\vartheta\geq k $, $\vartheta  N^{- \tau_1} \to 0 $ and $\vartheta  N^{(1-\alpha-4\tau_1)/4 } \to \infty $,
then for some finite $ t \geq 1 $, we have
\[  \Pr(S^{*} \subset \widehat S^{(t)}) \to 1~~\text{as}~~n \rightarrow \infty.  \] 
\end{theorem}

In practice, it is crucial to select an appropriate value of $k$ in order to ensure the property of sure screening.
For this, we consider the extended Bayesian information criterion (EBIC). 
Specifically,  for each $k \in \{ 1,2,\ldots,K\} $,  we carry out the proposed procedure and obtain $\widehat S_{[k]}$,
where $K$ is a prespecified constant.
Then we determine $k$ by minimizing 
\[  
\text{EBIC}(k) = \ell(\widehat{\be}_{[k]}) + \frac{1}{N} k (\ln N + 0.5 \ln p).
\]
Under mild regularity conditions, ~\cite{chen2012extended} established the selection consistency of EBIC under the ultrahigh dimensional generalized linear models. 
This, together with Theorem~\ref{thm: screening property} and Theorem~\ref{thm: DIHT screening property}, guarantees the selection consistency of our screening procedure. 

\section{Numerical Studies}\label{SRA}
\subsection{Simulation studies}
In this section, we evaluate the finite sample performance of the proposed method through simulation studies.
We compare the proposed distributed variable screening (DVS) method  with the aggregated correlation screening~\cite{JMLR:v21:19-537}. 
We consider the Pearson correlation, Kendall's $ \tau $ rank correlation, SIRS correlation and distance correlation (DC)
to perform the aggregated correlation screening; see~\cite{JMLR:v21:19-537} for more technical details.
The number of machines is taken as $m = 10, 30$ and $ 50$.
In addition, we set $N=3000$ and $p=6000$.

We assess the performance of the screening methods considered here based on $T = 100$ simulation replications.
For any specific method, let $ \widehat{S}(t)$ denote the model selected in the $ t $th replication.
The following criteria are used to compare the performance of different methods:
\begin{itemize}
    \item[SC]: the proportion of the submodel that contains all active covariates, defined as $\text{SC}= T^{-1} \sum_{t=1}^{T} \mathbb{I} ( S^{*} \subset \widehat{S}(t) )$;
    \item[PSR]: the positive selection rate, defined as $\text{PSR} = T^{-1} \sum_{t=1}^{T}\Vert S^{*} \cap \widehat{S}(t)  \Vert_{0}/\|S^{*} \|_{0}$;
    \item[FDR]: the false discovery rate, defined as $\text{FDR} =T^{-1} \sum_{t=1}^{T}\Vert\widehat{S}(t) \setminus S^{*} \Vert_{0}/\| \widehat{S}(t) \|_{0}$;
    \item[AMS]: the average of the selected model size, defined as $\text{AMS} =T^{-1}\sum_{t=1}^{T}  \text{card}(\widehat{S}(t))$; 
    \item[CF]:  the proportion of correctly fitted models, defined as $\text{CF}=T^{-1}\sum_{t=1}^{T}\mathbb{I}(\widehat{S}(t) = S^*)$. 
\end{itemize}
The criterion SC evaluates the sure screening property of the considered methods.
The PSR and FDR depict two different aspects of a selection
result: a high PSR indicates that the most relevant features are identified, 
whereas a low FDR indicates that a few irrelevant features are misselected.
A powerful variable screening method should guarantee that
SC, PSR and CF are close to one, AMS is close to the true model size, and FDR is close to zero.

\noindent\textbf{Example 1} (Linear regression). 
We generate $ y_{ij} $ from
$y_{ij} = \boldsymbol{X}_{ij}^{\top} \be^* + \varepsilon_{ij},$
where $ \varepsilon_{ij} $ are independent random variables following the standard normal distribution.
For the covariate $\boldsymbol{X}_{ij}$, we consider the following two settings:

Case 1.1: We first generate $ \boldsymbol{Z}_{ij}=(z_{ij,1},\dots,z_{ij,p})^\top $ and $ \boldsymbol{W}_{ij}=(w_{ij,1},\dots,w_{ij,p})^\top$  
independently from a standard multivariate normal distribution. 
Then, we obtain $\boldsymbol{X}_{ij}=(x_{ij,1},\dots,x_{ij,p})^\top$ with $ x_{ij,l} = (z_{ij,l} + w_{ij,l}) / \sqrt{2} $ for every $ 1 \leq l \leq 5 $ 
and $ x_{ij,l} = (z_{ij,l} + \sum_{l^{\prime}=1}^{5}  z_{ij,l^{\prime}} )/2 $  for every $ 5 < l \leq p $. 
Moreover, we set $ \be^* = (2, 4, 6, 8, 10, 0, \ldots, 0)^\top$.
This example is motivated by~\cite{wang2009forward}, the correlation coefficient of $ x_{ij,1} $ and $ y_{ij} $ is substantially smaller than that of $ x_{ij,l} $ and $ y_{ij} $ for every $ l > 5 $. 
Consequently, identifying $ x_{ij,1} $ as a relevant predictor poses an exceptionally challenging task.

Case 1.2: The covariates $ \boldsymbol{X}_{ij}~(1 \leq i \leq m, 1 \leq j \leq n)$ are generated from a multivariate normal distribution $ N(\boldsymbol{0}, \boldsymbol{\Sigma}_{i}) $, 
where $ \boldsymbol{\Sigma}_{i}=( \upsilon_{i}^{\vert s-t \vert}: 1\le s,t\le p) $ and $ \upsilon_{i} $ is generated from a uniform distribution on $(0.2, 0.3) $.
Moreover, we set  $ \be^* = (0.25, -0.5, 1, 0.3, -0.2, 0, \ldots, 0)^\top$.
Under these settings, the covariates $\boldsymbol{X}_{ij}$ are independent but are not identically distributed on different machines.

The results are summarized in Table \ref{tab:linear regression results},
which indicates that the proposed procedure outperforms the marginal variable screening method studied in~\cite{JMLR:v21:19-537}.
Specifically, under Case 1.1, the value of AMS for the marginal variable screening methods, such as the Pearson, Kendall and SIRS, is about 6000,
which indicates that they fail to identify the true model.
In addition, the value of AMS for DC is one, leading to a low SC.
However, the value of AMS for the proposed method is less than 15 and the value of SC is close to one,
which indicates that the proposed method successfully identifies all of the relevant covariates with a moderate model size.
Under Case 1.2, our method remains effective with $m=10$. For instance, 
the value of PSR approaches 0.97, which is close to one, indicating a high degree of accuracy. 
Similarly, the FDR value approaches 0.05, which is close to zero, indicating a low rate of false discoveries.
Finally, as the number of machines $m$ increases from 10 to 50, the performance of all methods tends to deteriorate to some extent. 
Nevertheless, our method consistently outperforms the marginal screening procedures.

\begin{table}[htbp]
	\centering
	\caption{Simulation results of two different settings under linear regression model.}
	\label{tab:linear regression results}

     \resizebox{0.75\textwidth}{!}{
	 \begin{tabular}{c|ccccc|ccccc}
	  \toprule
	  \multicolumn{1}{c}{\multirow{2}{*}{Method}} & \multicolumn{5}{c}{Case 1.1}   & \multicolumn{5}{c}{Case 1.2} \\

	  \cmidrule(lr){2-6} \cmidrule(lr){7-11}   
       \multicolumn{1}{c}{}  & SC & CF & AMS & PSR & \multicolumn{1}{c}{ FDR }        & SC & CF & AMS & PSR & FDR   \\

    \hline
    & \multicolumn{5}{c|}{$m = 10 $ } & \multicolumn{5}{c}{$  m = 10 $ } \\
	Pearson & 1.00  & 0.00  & 6000.00  & 1.00  & 1.00      & 0.81  & 0.00  & 9.24  & 0.78  & 0.44  \\
        Kendall & 1.00  & 0.00  & 6000.00  & 1.00  & 1.00      & 0.79  & 0.00  & 10.13  & 0.76  & 0.47  \\
        SIRS  & 1.00  & 0.00  & 6000.00  & 1.00  & 1.00        & 0.71  & 0.00  & 8.97  & 0.61  & 0.51  \\
        DC    & 0.00  & 0.00  & 1.00  & 0.20  & 0.00          & 0.73  & 0.00  & 9.26  & 0.53  & 0.55  \\
        DVS  & 1.00  & 0.35  & 5.93  & 1.00  & 0.14       & 0.98  & 0.80  & 5.90  & 0.97  & 0.05  \\

    \hline
	& \multicolumn{5}{c|}{$ m = 30 $ } & \multicolumn{5}{c}{$ m = 30 $ } \\
	Pearson & 1.00  & 0.00  & 6000.00  & 1.00  & 1.00       & 0.72  & 0.00  & 10.06 & 0.71  & 0.52  \\
        Kendall & 1.00  & 0.00  & 6000.00  & 1.00  & 1.00       & 0.69  & 0.00  & 11.22 & 0.62  & 0.53  \\
        SIRS  & 1.00  & 0.00  & 6000.00  & 1.00  & 1.00          & 0.63  & 0.00  & 17.36  & 0.51  & 0.62  \\
        DC    & 0.00  & 0.00  & 1.00     & 0.20  & 0.00        & 0.61  & 0.00  & 18.61  & 0.46  & 0.69  \\
        DVS  & 1.00  & 0.12  & 10.47  & 1.00  & 0.44      & 0.86  & 0.69  & 7.25  & 0.85  & 0.31  \\

    \hline
	& \multicolumn{5}{c|}{$ m = 50 $ } & \multicolumn{5}{c}{$  m = 50 $ } \\
	Pearson & 1.00  & 0.00  & 6000.00  & 1.00  & 1.00       & 0.69  & 0.00  & 12.57 & 0.68  & 0.63  \\
        Kendall & 1.00  & 0.00  & 6000.00  & 1.00  & 1.00       & 0.67  & 0.00  & 13.36 & 0.57  & 0.67  \\
        SIRS  & 0.80  & 0.00  & 5999.80  & 0.96  & 1.00         & 0.59  & 0.00  & 20.61 & 0.43  & 0.78  \\
        DC    & 0.00  & 0.00  & 1.00     & 0.20  & 0.00         & 0.54  & 0.00  & 23.25 & 0.41  & 0.80  \\
        DVS  & 1.00  & 0.07  & 13.17  & 1.00  & 0.52      & 0.79  & 0.52  & 12.11  & 0.77  & 0.43  \\

	\bottomrule
	\end{tabular}  }	
\end{table}

\noindent\textbf{Example 2}  (Logistic regression).  
We generate a binary response $ y_{ij} $ from a Bernoulli distribution with successful probability $\text{logit}(\boldsymbol{X}_{ij}^{\top} \be^*)$,
where $\text{logit}(x)=\exp(x)/[1+\exp(x)]$ and $\be^* = (0, 1.5, 0, 2, 0,-0.6, 0, \ldots, 0)^\top$.
For the covariate $\boldsymbol{X}_{ij}$, we consider the following two settings:

Case 2.1:  The covariates $ \boldsymbol{X}_{ij}~(1 \leq i \leq m, 1 \leq j \leq n)$ are independently generated from the standard normal distribution $N(\boldsymbol{0}, \mathbf{I})$,
where $\mathbf{I}\in \mathbb{R}^p $ denotes the identity matrix. 
Under Case 2.1, the covariates $\boldsymbol{X}_{ij}$ are independently and identically distributed random vectors on all machines.

Case 2.2: The covariates $ \boldsymbol{X}_{ij}~(1 \leq i \leq m, 1 \leq j \leq n)$ are generated as in Case 1.2.

\begin{table}[htbp]
	\centering
	\caption{Simulation results of two different settings under logistic regression model.}
	\label{tab:logistic regression results}

	\resizebox{0.75\textwidth}{!}{
	\begin{tabular}{c|ccccc|ccccc}
	\toprule
	\multicolumn{1}{c}{\multirow{2}{*}{Method}} & \multicolumn{5}{c}{Case 2.1}   & \multicolumn{5}{c}{Case 2.2} \\

	\cmidrule(lr){2-6} \cmidrule(lr){7-11}   
    \multicolumn{1}{c}{}  & SC & CF & AMS & PSR & \multicolumn{1}{c}{ FDR }        & SC & CF & AMS & PSR & FDR   \\  

    \hline
	& \multicolumn{5}{c|}{$ m = 10 $ } & \multicolumn{5}{c}{$m = 10 $ } \\
	Pearson & 1.00  & 0.08  & 14.62  & 1.00  & 0.62     & 0.83  & 0.00  & 11.28  & 0.86  & 0.67  \\
    Kendall & 1.00  & 0.06  & 13.58  & 1.00  & 0.81      & 0.85  & 0.00  & 7.23  & 0.31  & 0.68  \\
    SIRS  & 0.98  & 0.05  & 17.95  & 0.97  & 0.67        & 0.72  & 0.00  & 9.14  & 0.82  & 0.64  \\
    DC    & 0.92  & 0.06  & 14.06  & 0.94  & 0.65        & 0.67  & 0.01  & 9.99  & 0.83  & 0.66  \\
    DVS & 1.00  & 0.97  & 3.27  & 1.00  & 0.06       & 0.97  & 0.73  & 3.57  & 0.96  & 0.06  \\   

	\hline
	& \multicolumn{5}{c|}{$ m = 30 $ } & \multicolumn{5}{c}{$ m = 30 $ } \\
	Pearson & 0.94  & 0.08  & 16.32  & 0.92  & 0.67       & 0.76  & 0.00  & 13.22  & 0.81  & 0.69  \\
        Kendall & 0.91  & 0.09  & 17.26  & 0.92  & 0.78       & 0.75  & 0.00  & 19.53  & 0.30  & 0.75  \\
        SIRS  & 0.94  & 0.04  & 19.31  & 0.82  & 0.71         & 0.62  & 0.00  & 21.57  & 0.64  & 0.77  \\
        DC    & 0.86  & 0.11  & 22.58  & 0.75  & 0.73         & 0.59  & 0.00  & 27.92  & 0.55  & 0.81  \\
        DVS  & 0.95  & 0.91  & 4.68  & 0.93  & 0.21       & 0.89  & 0.62  & 6.14  & 0.86  & 0.32  \\  

	\hline
	& \multicolumn{5}{c|}{$ m = 50 $ } & \multicolumn{5}{c}{$ m = 50 $ } \\
	Pearson & 0.89  & 0.09  & 16.32  & 0.86  & 0.65         & 0.71  & 0.00  & 15.26  & 0.63  & 0.74  \\
        Kendall & 0.86  & 0.07  & 19.62  & 0.83  & 0.78         & 0.62  & 0.00  & 17.28  & 0.29  & 0.78  \\
        SIRS  & 0.67  & 0.03  & 25.44  & 0.67  & 0.76           & 0.52  & 0.00  & 31.36  & 0.62  & 0.81  \\
        DC    & 0.58  & 0.08  & 28.91  & 0.52  & 0.82           & 0.51  & 0.00  & 36.17  & 0.49  & 0.86  \\
        DVS  & 0.90  & 0.81  & 7.50  & 0.87  & 0.55         & 0.72  & 0.54  & 10.26  & 0.72  & 0.58  \\

	\bottomrule
	\end{tabular}  }
\end{table}

The results are summarized in Table \ref{tab:logistic regression results}.
We observe that under Case 2.1 with $m=10$ and $30$,
all methods show that the values of SC and PSR approach 1.
This indicates that the true model $ S^* $ is retained within the selected model with a high probability for all methods. 
However, the proposed method outperforms the other four methods in terms of FDR, resulting in a smaller value for AMS.
This suggests that the proposed method successfully identifies all the relevant covariates 
while having a smaller model size compared to the marginal variable screening procedures.
When $m=50$, the DC method shows a relatively poor performance with a SC of 0.58.
In addition, the proposed method still outperforms the marginal screening method under Case 2.2,
and the findings are similar to that of Case 1.2 in Example 1.

\noindent\textbf{Example 3}  (Poisson regression). 
The response $ y_{ij} $ is generated from a Poisson distribution with parameter $\lambda_{ij}=\exp(\boldsymbol{X}_{ij}^{\top} \be^*)$, 
where $ \be^* =(0, 0.8, -0.6, 0, 0.5, 0, \ldots, 0)^\top$.
The covariate $\boldsymbol{X}_{ij}$ are generated as in Example 2.

The results are summarized in Table \ref{tab:poisson regression results}.
It indicates that the proposed method still outperforms the marginal screening method,
and the findings are similar to that of Example 2.

\begin{table}[htbp]
	\centering
	\caption{Simulation results of two different settings under Poisson regression model.}
	\label{tab:poisson regression results}
	\resizebox{0.75\textwidth}{!}{
	\begin{tabular}{c|ccccc|ccccc}
		\toprule
	  \multicolumn{1}{c}{\multirow{2}{*}{Method}} & \multicolumn{5}{c}{Case 2.1}   & \multicolumn{5}{c}{Case 2.2} \\
	  \cmidrule(lr){2-6} \cmidrule(lr){7-11}   
       \multicolumn{1}{c}{}  & SC & CF & AMS & PSR & \multicolumn{1}{c}{ FDR }        & SC & CF & AMS & PSR & FDR   \\   
	  \hline
	    & \multicolumn{5}{c|}{$ m = 10 $ } & \multicolumn{5}{c}{$ m = 10 $ } \\
	   Pearson & 1.00  & 0.09  & 13.68  & 1.00  & 0.60         & 0.87  & 0.02  & 11.40  & 0.77  & 0.62  \\
         Kendall & 1.00  & 0.07  & 12.05  & 1.00  & 0.74          & 0.82  & 0.00  & 6.26  & 0.73  & 0.73  \\
         SIRS  & 0.96  & 0.09  & 17.61  & 0.95  & 0.66            & 0.79  & 0.00  & 10.09  & 0.72  & 0.65  \\
         DC    & 0.92  & 0.06  & 17.15  & 0.91  & 0.72            & 0.78  & 0.00  & 8.43  & 0.75  & 0.74  \\
        DVS & 1.00  & 0.98  & 4.12  & 1.00  & 0.05            & 0.97  & 0.82  & 5.07  & 0.97  & 0.26  \\	
	  \hline
	    & \multicolumn{5}{c|}{$m = 30 $ } & \multicolumn{5}{c}{$m = 30 $ } \\
	   Pearson & 0.93  & 0.07  & 14.13  & 0.92  & 0.65         & 0.79  & 0.01  & 12.06  & 0.71  & 0.70  \\
          Kendall & 0.91  & 0.08  & 13.61  & 0.91  & 0.76         & 0.73  & 0.00  & 10.25  & 0.69  & 0.75  \\
         SIRS  & 0.91  & 0.07  & 15.27  & 0.83  & 0.69            & 0.71  & 0.00  & 17.69  & 0.65  & 0.73  \\
         DC    & 0.83  & 0.05  & 17.31  & 0.71  & 0.75            & 0.62  & 0.00  & 25.11  & 0.62  & 0.78  \\
         DVS & 0.94  & 0.94  & 4.29  & 0.96  & 0.21           & 0.85  & 0.68  & 6.59  & 0.90  & 0.35  \\
	  \hline
	    & \multicolumn{5}{c|}{$ m = 50 $ } & \multicolumn{5}{c}{$ m = 50 $ } \\
	   Pearson & 0.87  & 0.14  & 12.05  & 0.86  & 0.58         & 0.71  & 0.00  & 14.53  & 0.63  & 0.73  \\
          Kendall & 0.85  & 0.06  & 20.37  & 0.79  & 0.82         & 0.66  & 0.00  & 15.27  & 0.59  & 0.79  \\
          SIRS  & 0.68  & 0.03  & 23.55  & 0.57  & 0.73           & 0.62  & 0.00  & 27.21  & 0.52  & 0.82  \\
          DC    & 0.61  & 0.03  & 26.71  & 0.49  & 0.79           & 0.53  & 0.00  & 29.33  & 0.41  & 0.87  \\
          DVS & 0.89  & 0.87  & 5.21  & 0.86  & 0.55          & 0.77  & 0.56  & 8.16  & 0.71  & 0.49  \\

	  \bottomrule
	  \end{tabular}	  }
  \end{table}

\subsection{Real Data Analysis}
In this section, we apply the proposed method to analyze the BASEHOCK dataset. 
It is a text dataset specifically designed for binary classification tasks,
and consists of 1993 samples with 4862 covariates.
Therefore, we consider a logistic regression model for this data analysis.
The data can be downloaded from \url{https://jundongl.github.io/scikit-feature/datasets.html}. 

In order to evaluate the performance of our method, 
we randomly partitioned the dataset into $m$ subgroups. 
The value of $m$ represents the number of machines used in the process, 
and we specifically set $m$ to be equal to 1, 10, 20, and 30. 
Note that the case where $m=1$ serves as the benchmark against which we compare the performance of other cases.
To reduce the potential impacts of sample partitioning on performance, we perform 100 replications of our method.
Specifically, at the $j$th replication, we divide the dataset into $m$ subgroups 
and then present the proposed method. The selected model in this replication is denoted as $\widehat{S}(t)$.

\begin{table}[htbp]
	\centering
	\caption{Real data analysis: The average size of the selected models from 100 replications.}
	\label{tab:real data results}
\resizebox{0.5\textwidth}{!}{
	\begin{tabular}{ccccc}
		\toprule
	& $ m=1 $   & $ m=10 $  & $ m=20 $  & $ m=30 $  \\	    
		\midrule
		Pearson & 28.31 & 31.82 & 30.29 & 34.31 \\
		Kendall & 46.55 & 47.93 & 45.02 & 50.46 \\
		SIRS  & 62.17 & 67.25 & 65.32 & 71.29 \\
		DC    & 122.05 & 131.71 & 129.61 & 152.06 \\
		DVS & 13.85 & 12.97 & 11.76 & 10.64 \\
    \bottomrule				
\end{tabular}  
}
\end{table}

Table~\ref{tab:real data results} presents the average size of selected models from the 100 replications,
which is calculated as $ 100^{-1} \sum_{t=1}^{100}  \text{card}(\widehat{S}(t))$.
Table~\ref{tab:BASEHOCK results} reports the selected covariates for the proposed methods.
The results demonstrate that the proposed method selects a smaller number of covariates compared to the marginal screening method. 
Furthermore, the proposed method exhibits consistent stability when $m$ ranges from 10 to 30.

\begin{table}[htbp]
	\centering
	\caption{Real data analysis: The selected covariates for the proposed method.}
	\label{tab:BASEHOCK results}
    \resizebox{0.5\textwidth}{!}{
	\begin{tabular}{ccccc}
		\toprule		    
        Covariates &   $ m=1 $   & $ m=10 $   & $ m=20 $  & $ m=30 $    \\
        \midrule

	   X357  & $ \surd $ & $ \surd $  & $ \surd $  & $ \surd $  \\
		X3303 & $ \surd $ & $ \surd $ & $ \surd $  & $ \surd $  \\
		X4681 & $ \surd $ & $ \surd $   &       &  \\
		X3083 & $ \surd $  &       & $ \surd $   & $ \surd $  \\
		X1309 & $ \surd $ & $ \surd $ & $ \surd $  & $ \surd $  \\
		X4357 & $ \surd $  & $ \surd $ & $ \surd $  & $ \surd $  \\
		X1860 & $ \surd $  & $ \surd $ &       &  \\
		X1404 & $ \surd $  &       & $ \surd $  &  \\
		X847  & $ \surd $  &       &       & $ \surd $  \\
		X769  & $ \surd $  & $ \surd $  & $ \surd $ &  \\
		X1273 & $ \surd $  &       & $ \surd $  & $ \surd $  \\
		X1029 & $ \surd $  & $ \surd $  & $ \surd $  &  \\
		X3282 & $ \surd $  & $ \surd $  &       &  \\
		X1794 & $ \surd $  & $ \surd $  &       &  \\
		X2826 & $ \surd $  &       &       &  \\
		X2966 &       & $ \surd $  & $ \surd $  & $ \surd $  \\
		X2006 &       & $ \surd $  & $ \surd $  & $ \surd $  \\
		X1800 &       &       &       & $ \surd $   	\\	
		\bottomrule
	\end{tabular}
}
\end{table}

\section{Conclusion}\label{CR}
In this article, we developed a distributed variable screening method for the generalized linear regression models.
This method operated using a sparsity-restricted surrogate likelihood estimator. 
We also established its sure screening property. 
Additionally, we evaluated the performance of the proposed method through simulation studies and real data analysis.

It is worth noting that our method is specifically designed for linear models. 
However, it has the potential to be extended to nonlinear models or semiparametric models 
such as single index models, partial linear models, and Cox models. 
This extension would be interesting and will be a subject of our future research.

\section*{Acknowledgments}
Qu's research was partially supported by the National Natural Science Foundation of China (No.12001219). 
Li's research was partially supported by the National Natural Science Foundation of China (No.62377019). 
Sun's research was partially supported by the National Natural Science
Foundation of China (No.12171463). 

\bibliographystyle{asa}
\bibliography{references}

\begin{thebibliography}{21}
\newcommand{\enquote}[1]{``#1''}
\expandafter\ifx\csname natexlab\endcsname\relax\def\natexlab#1{#1}\fi

\bibitem[{Bickel et~al.(2009)Bickel, Ritov, and Tsybakov}]{bickel2009simultaneous}
Bickel, P.~J., Ritov, Y., and Tsybakov, A.~B. (2009), \enquote{Simultaneous analysis of Lasso and Dantzig selector,} \textit{The Annals of Statistics}, 37, 1705 -- 1732.

\bibitem[{Candes and Tao(2007)}]{candes2007dantzig}
Candes, E. and Tao, T. (2007), \enquote{The Dantzig selector: Statistical estimation when $p$ is much larger than $n$,} \textit{The Annals of Statistics}, 35, 2313--2351.

\bibitem[{Chen and Chen(2012)}]{chen2012extended}
Chen, J. and Chen, Z. (2012), \enquote{Extended BIC for small-$n$-large-$P$ sparse GLM,} \textit{Statistica Sinica}, 22, 555--574.

\bibitem[{Chen et~al.(2014)Chen, Ge, Wang, and Ye}]{chen2014complexity}
Chen, X., Ge, D., Wang, Z., and Ye, Y. (2014), \enquote{Complexity of unconstrained $ L_2 $-$ L_p $ minimization,} \textit{Mathematical Programming}, 143, 371--383.

\bibitem[{Fan and Lv(2008)}]{fan2008sure}
Fan, J. and Lv, J. (2008), \enquote{Sure independence screening for ultrahigh dimensional feature space,} \textit{Journal of the Royal Statistical Society: Series B (Statistical Methodology)}, 70, 849--911.

\bibitem[{Fan and Lv(2010)}]{fan2010selective}
--- (2010), \enquote{A selective overview of variable selection in high dimensional feature space,} \textit{Statistica Sinica}, 20, 101--148.

\bibitem[{Fan and Song(2010)}]{fan2010sure}
Fan, J. and Song, R. (2010), \enquote{Sure independence screening in generalized linear models with NP-dimensionality,} \textit{The Annals of Statistics}, 38, 3567 -- 3604.

\bibitem[{Gao et~al.(2022)Gao, Liu, Wang, Wang, Yan, and Zhang}]{gao2022review}
Gao, Y., Liu, W., Wang, H., Wang, X., Yan, Y., and Zhang, R. (2022), \enquote{A review of distributed statistical inference,} \textit{Statistical Theory and Related Fields}, 6, 89--99.

\bibitem[{Hao et~al.(2021)Hao, Qu, Kong, Sun, and Zhu}]{hao2021optimal}
Hao, M., Qu, L., Kong, D., Sun, L., and Zhu, H. (2021), \enquote{Optimal minimax variable selection for large-scale matrix linear regression model,} \textit{Journal of Machine Learning Research}, 22, 1--39.

\bibitem[{Jordan et~al.(2019)Jordan, Lee, and Yang}]{jordan2019communicationefficient}
Jordan, M.~I., Lee, J.~D., and Yang, Y. (2019), \enquote{Communication-efficient distributed statistical inference,} \textit{Journal of the American Statistical Association}, 114, 668--681.

\bibitem[{Li et~al.(2012)Li, Zhong, and Zhu}]{li2012feature}
Li, R., Zhong, W., and Zhu, L. (2012), \enquote{Feature screening via distance correlation learning,} \textit{Journal of the American Statistical Association}, 107, 1129--1139.

\bibitem[{Li et~al.(2020)Li, Li, Xia, and Xu}]{JMLR:v21:19-537}
Li, X., Li, R., Xia, Z., and Xu, C. (2020), \enquote{Distributed feature screening via componentwise debiasing,} \textit{Journal of Machine Learning Research}, 21, 1--32.

\bibitem[{Li and Xu(2023)}]{li2023feature}
Li, X. and Xu, C. (2023), \enquote{Feature screening with conditional rank utility for big-data classification,} \textit{Journal of the American Statistical Association}, 1--11.

\bibitem[{Natarajan(1995)}]{natarajan1995sparse}
Natarajan, B.~K. (1995), \enquote{Sparse approximate solutions to linear systems,} \textit{SIAM journal on computing}, 24, 227--234.

\bibitem[{Nesterov(2004)}]{nesterov2003introductory}
Nesterov, Y. (2004), \textit{Introductory lectures on convex optimization: A basic course}, Springer.

\bibitem[{Wang(2009)}]{wang2009forward}
Wang, H. (2009), \enquote{Forward regression for ultra-high dimensional variable screening,} \textit{Journal of the American Statistical Association}, 104, 1512--1524.

\bibitem[{Xu and Chen(2014)}]{xu2014sparse}
Xu, C. and Chen, J. (2014), \enquote{The sparse MLE for ultrahigh-dimensional feature screening,} \textit{Journal of the American Statistical Association}, 109, 1257--1269.

\bibitem[{Yang et~al.(2016)Yang, Yu, Li, and Buu}]{yang2016feature}
Yang, G., Yu, Y., Li, R., and Buu, A. (2016), \enquote{Feature screening in ultrahigh dimensional Cox's model,} \textit{Statistica Sinica}, 26, 881.

\bibitem[{Zhou et~al.(2023)Zhou, Gong, and Xiang}]{Zhou2024}
Zhou, L., Gong, Z., and Xiang, P. (2023), \enquote{Distributed computing and inference for big data,} \textit{Annual Review of Statistics and Its Application}.

\bibitem[{Zhou et~al.(2020)Zhou, Zhu, Xu, and Li}]{zhou2020model}
Zhou, T., Zhu, L., Xu, C., and Li, R. (2020), \enquote{Model-free forward screening via cumulative divergence,} \textit{Journal of the American Statistical Association}, 115, 1393--1405.

\bibitem[{Zhu et~al.(2011)Zhu, Li, Li, and Zhu}]{zhu2011model}
Zhu, L., Li, L., Li, R., and Zhu, L. (2011), \enquote{Model-free feature screening for ultrahigh-dimensional data,} \textit{Journal of the American Statistical Association}, 106, 1464--1475.

\end{thebibliography}

\newpage




\appendix
\section*{Appendix}
\renewcommand\theequation{A.\arabic{equation}}
\begin{proof}[Proof of Theorem~\ref{thm: IHT}]	
By the definition of $\ell(\be)$, an application of the Taylor's expansion yields
\begin{align} \label{Theorem1:eq0} 
    \ell(\boldsymbol{\gamma}) 
    &= \ell(\be) + \langle \boldsymbol{\gamma} - \be, \nabla \ell(\be) \rangle + \frac{1}{2}(\boldsymbol{\gamma} - \be)^{\top} \nabla^2 \ell (\overline{\be}) (\boldsymbol{\gamma} - \be)   \nonumber  \\   
    &= \ell(\be) + \langle \boldsymbol{\gamma} - \be, \nabla \ell(\be) \rangle +\frac{1}{2}(\boldsymbol{\gamma} - \be)^{\top} \left\lbrace \frac{1}{n} \sum_{j=1}^{n} \left(  \boldsymbol{X}_{1j} b^{\prime\prime} ( \boldsymbol{X}_{1j}^{\top} \overline{\be})  \boldsymbol{X}_{1j}^{\top} \right)   \right\rbrace  (\boldsymbol{\gamma} - \be)   \nonumber  \\
    &= \ell(\be) + \langle \boldsymbol{\gamma} - \be, \nabla \ell(\be) \rangle +\frac{1}{2}(\boldsymbol{\gamma} - \be)^{\top} \left\lbrace \frac{1}{n} \mathbb{X}_1^{\top} \text{diag}\{b^{\prime\prime} ( \boldsymbol{X}_{11}^{\top} \overline{\be}), \ldots, b^{\prime\prime} ( \boldsymbol{X}_{1n}^{\top} \overline{\be})\}   \mathbb{X}_1 \right\rbrace  (\boldsymbol{\gamma} - \be)   \nonumber    \\
    &\leq \ell(\be) + \langle \boldsymbol{\gamma} - \be, \nabla \ell(\be) \rangle +\frac{1}{2n}\mu  (\boldsymbol{\gamma} - \be)^{\top}   \mathbb{X}_1^{\top} \mathbb{X}_1   (\boldsymbol{\gamma} - \be)   \nonumber  \\
    &\leq \ell(\be) + \langle \boldsymbol{\gamma} - \be, \nabla \ell(\be) \rangle + \frac{1}{2}\vartheta \Vert \boldsymbol{\gamma} - \be \Vert_2^2 =  h(\boldsymbol{\gamma}; \be),   
  \end{align}
 where $ \overline{\be} $ lies between $ \boldsymbol{\gamma} $ and $ \be $.
 The last two inequalities hold due to the facts that $\mu=\max_{\theta\in\bar\Theta} b^{\prime \prime}(\theta)$ and $\vartheta>\rho_1\mu/n$.
 This implies that 
  \begin{align}  \label{Theorem1:eq1}
      \ell ( \widehat{ \be }^{(t+1)} )  \leq   h ( \widehat{ \be }^{(t+1)}; \widehat{ \be }^{(t)} ). 
  \end{align}
  In addition, by the definition of $ \widehat{ \be }^{(t+1)}$, we have 
  \begin{align} \label{Theorem1:eq2}
  h ( \widehat{ \be }^{(t+1)}; \widehat{ \be }^{(t)} )  \leq h ( \widehat{ \be }^{(t)}; \widehat{ \be }^{(t)} )  = \ell ( \widehat{ \be }^{(t)} ).
  \end{align}
Then by~\eqref{Theorem1:eq1} and~\eqref{Theorem1:eq2}, we obtain $ \ell ( \widehat{ \be }^{(t+1)} ) \leq \ell ( \widehat{ \be }^{(t)} ) $.
  This completes the proof.
\end{proof}

\begin{proof}[Proof of Theorem~\ref{thm: convergence}]
  By~\eqref{Theorem1:eq2} and some arguments similar to the proof of~\eqref{Theorem1:eq0}, we have
  \begin{align}  \label{pf2:eq1}
  \ell (\widehat{ \be }^{(t)}) 
  &\geq  h\big( \widehat{ \be }^{(t+1)}; \widehat{ \be }^{(t)} \big) \nonumber \\
  &= \ell( \widehat{ \be }^{(t)} ) + \langle \widehat{ \be }^{(t+1)} - \widehat{ \be }^{(t)}, \nabla \ell( \widehat{ \be }^{(t)} ) \rangle +\frac{\vartheta}{2} \Vert \widehat{ \be }^{(t+1)} - \widehat{ \be }^{(t)} \Vert_2^2  \nonumber\\
  &= \ell( \widehat{ \be }^{(t+1)} ) +\frac{\vartheta}{2} \Vert \widehat{ \be }^{(t+1)} - \widehat{ \be }^{(t)} \Vert_2^2 - \big[ \ell( \widehat{ \be }^{(t+1)} )- \ell(\widehat{ \be }^{(t)} ) -\langle \widehat{ \be }^{(t+1)} - \widehat{ \be }^{(t)}, \nabla \ell( \widehat{ \be }^{(t)} ) \rangle  \big] \nonumber \\
  &\ge \ell( \widehat{ \be }^{(t+1)} )  + \frac{1}{2} (\vartheta - \rho_1 \rho^{(t)}/n ) \Vert \widehat{ \be }^{(t+1)} - \widehat{ \be }^{(t)} \Vert_2^2,
  \end{align}
  where $ \rho^{(t)}=\max_{1\le i\le n} b^{\prime\prime} ( \boldsymbol{X}_{1i}^{\top} \overline{\be})$ 
  and $\overline{\be}$ is between $\widehat\be^{(t)}$ and $\widehat\be^{(t+1)}$. 
  
Note that $ \ell(\be) $ is bounded. Therefore, $ \ell (\widehat{ \be }^{(t)} ) $ has a limiting point within the feasible region $(\|\be\|_0\le k)$.
  Based on this fact, we show that there exists a subsequence of $\widehat{ \be }^{(t)} $ that is convergent.
  Since $\rho^{(t)} \le \mu $ for all $t\ge 1$, there exists a subsequence $ \mathcal{S} $ of $\{t\}_{t=1}^{\infty}$ 
  such that $ \rho^{[l]} \rightarrow  \widetilde{\rho}~(l\in\mathcal{S})$. 
In addition, $ \ell (\widehat{\be}^{(l)}) $ is convergent, due to the fact that $ \ell (\widehat{\be}^{(l)}) $ is bounded and decreasing.
  These facts, together with \eqref{pf2:eq1}, 
  imply that $ \Vert \widehat{\be}^{(l+1)} - \widehat{\be}^{(l)} \Vert_2 \rightarrow 0 $ as $ l \in \mathcal{S} $ goes to infinity.
  Consequently, $\widehat{M}^{(l)}=\{1\le j\le p: \widehat\beta_j^{(l)}\neq 0\}$ is also convergent, where $l\in\mathcal{S}$.
  Since the support $\widehat{M}^{(l)}$ is a discrete sequence, 
  there exists a constant $ l^* \in \mathcal{S} $ such that 
  $\widehat{M}^{(l)}= \widehat{M}^{(l^*)}$ for all $ l \in \mathcal{S} $ satisfying $ l \ge l^* $. 
  Thus, Algorithm~\ref{DIHT algorithm} becomes a gradient descent algorithm on  $\widehat{M}^{(l)}$ for all $ l \in \mathcal{S} $ and $ l \ge l^* $. 
  Because a gradient descent algorithm for minimizing a convex function over a closed convex set yields a sequence of iterations that converges~\cite{nesterov2003introductory}, we conclude that the subsequence $ \{ \widehat{ \be }^{(t)}: t \in \mathcal{S} \}$ is convergent.
  
  Next, we prove the second part of Theorem \ref{thm: convergence}, that is, Algorithm~\ref{DIHT algorithm} stops in a finite number of steps.
  By~\eqref{pf2:eq1}, we obtain
  \begin{align*}
  \sum_{t=0}^K \ell( \widehat{ \be }^{(t)} ) - \ell( \widehat{ \be }^{(t+1)} )  &\geq \frac{1}{2}\big(\vartheta - \rho_1 \mu/n\big)  \sum_{t=0}^K \Vert \widehat{ \be }^{(t+1)} - \widehat{ \be }^{(t)} \Vert_2^2.
  \end{align*}
  This implies 
  \begin{align*}
    \min_{0 \le t \le K} \Vert \widehat{ \be }^{(t+1)} - \widehat{ \be }^{(t)} \Vert_2^2  
    \le \frac{2}{K (\vartheta - \rho_1 \mu/n)} \big( \ell(\widehat\be^{(0)}) - \ell(\widehat\be^{(K+1)}) \big).
  \end{align*}
  Let $ \ell (\be^{\star} ) $ be the limiting point of $ \ell( \widehat{ \be }^{(t)} ) $. Then we have
  \begin{align*}
  \min_{0 \le t \le K} \Vert \widehat{ \be }^{(t+1)} - \widehat{ \be }^{(t)} \Vert_2^2
    \le \frac{2}{K (\vartheta  - \rho_1 \mu/n)} \big( \ell(\widehat\be^{(0)}) - \ell (\be^{\star} )\big).
  \end{align*}
  Note that $ \ell(\be) $ is bounded.  Thus, if set $ K \gtrsim  1/(\varepsilon^2 (\vartheta  - \rho_1 \mu/n)) $, 
  then  for any $ \varepsilon >0 $, there exists some $ 0 \le \tilde t \le K $ 
  such that $ \Vert  \widehat{\be}^{(\tilde t +1)} - \widehat{\be}^{(\tilde t)} \Vert_2^2 \leq \varepsilon^2 $.
  This completes the proof.
  \end{proof}
  
  To prove Theorem~\ref{thm: screening property}, we need the following two lemmas.
  
  \begin{lemma} \label{lemma 1}
  Let $ \{Y_{i}, i=1,\ldots, N \} $ be independent random variables from exponential family distributions with natural parameters $ \theta_{i} \in \bar{\Theta} $. 
  Denote the mean and variance of $ Y_{i} $ as $ \mu_{i} $ and $ \sigma_{i}^{2} $. 
  If there exist some constants $ t_{N i}~(i=1, \ldots, N )$ such that 
  \[  \sum_{i=1}^{N} t_{N i}^{2} \sigma_{i}^{2}=1~~\text{and}~~\max _{1 \leq i \leq N}\left\{t_{N i}^{2} \right\} = O\left(N^{-1}\right),  \]  
  then  for some positive sequence $ d_{N}=o\left(N^{1 / 2}\right) $ and a sufficiently large $N$, we have
  \[  \Pr \left(\sum_{i=1}^{N} t_{N i} \left(Y_{i}-\mu_{i}\right)>d_{N}\right) \leq \exp \left(-d_{N}^{2} / 3\right).  \]  
  \end{lemma}
  \begin{proof}
  The proof of Lemma~\ref{lemma 1} can be found in~\cite{chen2012extended}.
\end{proof}

\begin{lemma} \label{lemma 2}
    Let $ \widetilde{ \be }$ be the Lasso estimator defined in \eqref{initial:Lasso}.
    If Conditions (C4) and (C5) hold with $ \lambda \rightarrow 0 $ as $ n \rightarrow \infty $,
    then with probability at least $ 1 - 2p \exp(-C_1 n \lambda^2) $, we have
    \[   \Vert  \widetilde{ \be } - \be^{*} \Vert_{2} \leq C_2 \lambda  q,    \]
    and
    \[  \Vert  \widetilde{ \be } - \be^{*} \Vert_{1} \leq C_2 \lambda q,  \]
    where $ C_1 $ and $ C_2 $ are some positive constants independent of $n$.
\end{lemma}

\begin{proof}[Proof of Lemma~\ref{lemma 2}]
  The proof follows from Theorem 7.2 in~\cite{bickel2009simultaneous}.
  By the definition of $ \widetilde{ \be } $, we have 	
  \begin{align}\label{lemma3:eq0}
    \mathcal{L}_{1}(\widetilde{ \be }) - \mathcal{L}_{1}(\be^{*})  \leq  \lambda(\Vert \be^* \Vert_{1} -\Vert \widetilde{ \be } \Vert_{1}).   
  \end{align}
  In addition, by the Taylor's expansion, we have
  \[  \mathcal{L}_{1}(\widetilde{ \be }) - \mathcal{L}_{1}(\be^{*}) = \left \langle \nabla\mathcal{L}_{1}(\be^*) ,  \boldsymbol{w} \right \rangle 
    + (1/2) \boldsymbol{w}^\top\nabla^2 \mathcal{L}_{1}(\overline{\be})\boldsymbol{w},  \]  
  where $ \boldsymbol{w} = \widetilde{ \be } - \be^{*} $ and $ \overline{\be} $ is between $ \widetilde{ \be } $ and $ \be^* $. 
  This, together with~\eqref{lemma3:eq0}, implies
  \begin{align} \label{algo:eq1}
  \boldsymbol{w}^\top (\nabla^2 \mathcal{L}_{1}(\overline{\be})) \boldsymbol{w}  
  &= -2\langle \nabla\mathcal{L}_{1}(\be^*),\boldsymbol{w} \rangle 
  + 2( \mathcal{L}_{1}(\widetilde{ \be }) - \mathcal{L}_{1}(\be^{*})) \nonumber \\
  &\leq 2 \Vert \nabla\mathcal{L}_{1}(\be^*) \Vert_{\infty}  \Vert \boldsymbol{w} \Vert_{1} + 2 \lambda( \Vert \be^* \Vert_{1} - \Vert \widetilde{ \be } \Vert_{1}). 
    \end{align}
  Define
    \[  \mathcal{S}_0 =\Big\{ \max_{ 1 \le l \le p} | (\nabla \mathcal{L}_{1}(\be^*))_l | \le \lambda/2\Big\},   \] 
  where $ (\nabla \mathcal{L}_{1} \left( \be^{*} \right))_l$ is the $ l $th component of $ \nabla \mathcal{L}_{1}( \be^{*}) $ and is defined as
    \[
      (\nabla \mathcal{L}_{1} \left( \be^{*} \right))_l = n^{-1}\sum_{j=1}^{n}  \left[  b^{\prime}(\boldsymbol{X}_{1j}^{\top} \be^{*} ) - y_{1j} \right] x_{1j,l}
      = n^{-1}\sum_{j=1}^{n}  \left[  \mathbb{E} (y_{1j}|\boldsymbol{X}_{1j}) - y_{1j} \right] x_{1j,l}.
    \]     
  Let $ t_{nj} = x_{1j,l} (  \sum_{j=1}^{n} x_{1j,l}^{2} \sigma_{1j}^{2})^{-1/2}$. It can be checked that $ \sum_{j=1}^{n} t_{nj}^{2} \sigma_{1j}^{2}=1 $. 
  Since $ b(\cdot) $ is twice continuously differentiable on $ \bar{\Theta} $, we have $ b^{\prime \prime}(\theta) \leq \mu $ for some constant $ \mu>0$. 	
  Thus, by Condition (C4), we have 
  \[  \max _{j}\{t_{nj}^{2}\} = O(n^{-1})~~ 
  \text{and}~~ \sum_{j=1}^{n} x_{1j,l}^{2} \sigma_{1j}^{2} \leq c_{2}^{2} n \mu^{2}.  \] 
  Then, Lemma~\ref{lemma 1} implies
  \begin{align}  \label{pr of complement}
    \Pr( (\mathcal{S}_0)_{c} ) 
    &\leq \sum_{l=1}^{p} \Pr(| (\nabla \mathcal{L}_{1}(\be^*))_l | >\lambda/2)   \nonumber \\ 
    &\leq  2p \exp(-C_1 n\lambda^{2}), 
  \end{align}
  where $C_1$ is a positive constant. 
  
  We next show that on the event $\mathcal{S}_0$, $\|\boldsymbol{w}\|_1\le 16q b_1^{-1}\lambda$,
  where $b_1$ is some positive constant. By~\eqref{algo:eq1}, we have 
    \[  \boldsymbol{w}^\top (\nabla^2 \mathcal{L}_{1}(\overline{\be})) \boldsymbol{w} \leq \lambda \Vert \boldsymbol{w} \Vert_{1} 
    + 2 \lambda( \Vert \be^* \Vert_{1} - \Vert \widetilde{ \be } \Vert_{1}),  \] 
    which implies that
    \begin{align} \label{lemma3:eq1}
    \boldsymbol{w}^\top\nabla^2 \mathcal{L}_{1}(\overline{\be})\boldsymbol{w} + \lambda  \Vert \boldsymbol{w} \Vert_{1}  
    &\leq 2 \lambda \bigg[ \sum_{i=1}^{p}\Big(  \vert \widetilde{\beta}_i - \beta^*_i \vert  +  \vert  \beta^*_i \vert - \vert \widetilde{\beta}_i \vert \Big)\bigg]\nonumber  \\
    &= 2 \lambda\bigg[ \sum_{i \in S^*}\Big (  \vert \widetilde{\beta}_i - \beta^*_i \vert  +  \vert  \beta^*_i \vert - \vert \widetilde{\beta}_i \vert \Big) \bigg]\nonumber  \\
    &\leq 4\lambda \sum_{i \in S^*}\vert \widetilde{\beta}_i - \beta^*_i \vert= 4 \lambda \Vert \boldsymbol{w}_{S^*} \Vert_{1}. 
    \end{align}
  Since $ \boldsymbol{w}^\top\nabla^2 \mathcal{L}_{1}( \overline{\be}) \boldsymbol{w} \geq 0 $, we obtain
  $\Vert \boldsymbol{w}_{S^*_c} \Vert_{1} \leq 3 \Vert \boldsymbol{w}_{S^*} \Vert_{1}.$
  This, together with (C5) and \eqref{lemma3:eq1}, implies that 
    \[  \Vert \boldsymbol{w}_{S^*} \Vert_{1}^2  \leq  q \Vert \boldsymbol{w}_{S^*} \Vert_{2}^2  
    \leq  q b_1^{-1}\big(\boldsymbol{w}^\top \mathcal{H}(\overline{\be}) \boldsymbol{w}\big) 
    \leq   4q b_1^{-1} \lambda \Vert \boldsymbol{w}_{S^*} \Vert_{1}. \]  
  That is, $\Vert \boldsymbol{w}_{S^*} \Vert_{1} \leq 4q b_1^{-1}\lambda.$
  In addition, it can be directly shown that
    \[  \Vert \boldsymbol{w} \Vert_{2} \leq  \Vert \boldsymbol{w} \Vert_{1} = \Vert \boldsymbol{w}_{S^*} \Vert_{1} + \Vert \boldsymbol{w}_{S^*_c} \Vert_{1}  
    \leq   4 \Vert \boldsymbol{w}_{S^*} \Vert_{1}    \leq   16q b_1^{-1}\lambda.   \]  
  This completes the proof.
\end{proof}
  
\begin{proof}[Proof of Theorem~\ref{thm: screening property}] 
  Let $ \mathcal{M}_{+}^{k} = \left\{ S: S^{*} \subset S , \Vert S \Vert_{0} \leq k \right\} $
  and $ \mathcal{M}_{-}^{k} = \left\{ S: S^{*} \not \subset S , \Vert S \Vert_{0} \leq k  \right\} $  
  denote the collections of over-fitted models and under-fitted models, respectively.
  To prove Theorem~\ref{thm: screening property}, it suffices to show that
  \begin{equation} \label{ssp:eq1}
  \Pr \Big\{ \min_{S \in \mathcal{M}_{-}^{k}} \ell(\widehat{\be}_{S}) \leq \max_{S \in \mathcal{M}_{+}^{k}} \ell (\widehat{\be}_{S}) \Big\} \rightarrow 0 ~~\text{as}~N\rightarrow \infty.
  \end{equation}
  For any $ S \in \mathcal{M}_{-}^{k} $, define $ S^{\prime}=S \cup S^{*}$. Thus, $S^{\prime} \in \mathcal{M}_{+}^{2 k} $. 
  Let $\be_{S^{\prime}}$ be any vector 
  such that $ \left\|\be_{S^{\prime}}-\be_{S^{\prime}}^{*}\right\|_2 =w_{1} N^{-\tau_{1}} $ for some $ w_{1}$ and $ \tau_{1}>0 $. 
  By the Taylor's expansion, we have 
  \begin{align*}
    \ell( \be_{S^{\prime}} )-\ell ( \be_{S^{\prime}}^{*} )
    & = (\be_{S^{\prime}}-\be_{S^{\prime}}^{*} )^{\top} \nabla \ell(\be_{S^{\prime}}^{*})
    +(1/2)(\be_{S^{\prime}}-\be_{S^{\prime}}^{*})^{\top} \mathcal H(\overline{\be}_{S^{\prime}})(\boldsymbol{\beta}_{S^{\prime}}-\boldsymbol{\beta}_{S^{\prime}}^{*}) \\
    & \geq (\boldsymbol{\beta}_{S^{\prime}}-\boldsymbol{\beta}_{S^{\prime}}^{*})^{\top} \nabla \ell(\boldsymbol{\beta}_{S^{\prime}}^{*})
    +(c_{1}/2)\| \boldsymbol{\beta}_{S^{\prime}}-\boldsymbol{\beta}_{S^{\prime}}^{*}\|_{2}^{2}                        \\
    & \geq -w_{1} N^{-\tau_{1}}\left\| \nabla \ell \left(\boldsymbol{\beta}_{S^{\prime}}^{*}\right)\right\|_{2}+\left(c_{1} / 2\right) w_{1}^{2} N^{-2 \tau_{1}},
  \end{align*}
  where $ \overline{\boldsymbol{\beta}}_{S^{\prime}} $ is between $ \boldsymbol{\beta}_{S^{\prime}} $ and $ \boldsymbol{\beta}_{S^{\prime}}^{*} $.
  The first inequality holds due to Condition (C3), while the last one holds by the Cauchy-Schwarz inequality.
Moreover,   
  \[  \nabla \ell ( \be_{S^{\prime}}^{*}) = \nabla \mathcal{L}( \widetilde{\be}_{S^{\prime}}) - \nabla \mathcal{L}( \be_{S^{\prime}}^{*}) + \nabla \mathcal{L}( \be_{S^{\prime}}^{*} )
  - \Big[\nabla\mathcal{L}_{1}(\widetilde{\be}_{S^{\prime}})  - \nabla \mathcal{L}_{1}(  \be_{S^{\prime}}^{*}) \Big].  \] 
Therefore, we have
  \begin{equation} \label{ssp:eq2}
   \begin{aligned} 
     \Pr\left\{ \ell \left( \be_{S^{\prime}}^{*} \right) - \ell \left( \be_{S^{\prime}} \right) \geq 0\right\}  
    \leq & \Pr\left\{\left\| \nabla \ell \left(\be_{S^{\prime}}^{*}\right)\right\|_{2} \geq \left(c_{1} w_{1} / 2\right) N^{-\tau_{1}} \right\}    \\ 
     \leq & \Pr\left\{  \left \Vert \nabla \mathcal{L}( \widetilde{\be}_{S^{\prime}}) - \nabla \mathcal{L} \left( \be_{S^{\prime}}^{*} \right) \right \Vert_2 \geq \left(c_{1} w_{1} / 6\right) N^{-\tau_{1}}  \right\}    \\ 
     &+ \Pr\left\{  \left \Vert \nabla \mathcal{L}_{1}( \widetilde{\be}_{S^{\prime}}) - \nabla \mathcal{L}_{1} \left( \be_{S^{\prime}}^{*} \right) \right \Vert_2 \geq \left(c_{1} w_{1} / 6\right) N^{-\tau_{1}}  \right\}     \\ 
     &+ \Pr\left\{  \left \Vert \nabla \mathcal{L} \left( \be_{S^{\prime}}^{*} \right) \right \Vert_2 \geq \left(c_{1} w_{1} / 6\right) N^{-\tau_{1}}  \right\}.   
   \end{aligned} 
  \end{equation}
  For the first term on the right-hand side of \eqref{ssp:eq2}, we have 
  \begin{align}  \label{ssp:eq:11}
  \nabla \mathcal{L} ( \widetilde{\be}_{S^{\prime}}) - \nabla \mathcal{L}( \be_{S^{\prime}}^{*}) 
   =& \frac{1}{N} \sum_{i=1}^{m} \sum_{j=1}^{n} \boldsymbol{X}_{ijS^{\prime}} 
   \left[  b^{\prime}(\boldsymbol{X}_{ijS^{\prime}}^{\top} \widetilde{\be}_{S^{\prime}})
   - b^{\prime}(\boldsymbol{X}_{ijS^{\prime}}^{\top}\be_{S^{\prime}}^{*})  \right]   \nonumber   \\
   =&  \frac{1}{N} \sum_{i=1}^{m} \sum_{j=1}^{n} \boldsymbol{X}_{ijS^{\prime}} \boldsymbol{X}_{ijS^{\prime}}^{\top} 
   b^{\prime\prime}(\boldsymbol{X}_{ijS^{\prime}}^{\top} \overline{\be})( \widetilde{\be}_{S^{\prime}} - \be_{S^{\prime}}^{*})   \nonumber   \\
  =& \frac{1}{N} \mathcal{X}_{S^\prime}^{\top}   \text{diag} \left(  b^{\prime\prime}(\boldsymbol{X}_{11S^{\prime}}^{\top} \overline{\be}) , \ldots, b^{\prime\prime}(\boldsymbol{X}_{1nS^{\prime}}^{\top} \overline{\be} \right) , \ldots, b^{\prime\prime}(\boldsymbol{X}_{mnS^{\prime}}^{\top} \overline{\be} ))\mathcal{X}_{S^\prime}  \big( \widetilde{\be}_{S^{\prime}} - \be_{S^{\prime}}^{*} \big),
  \end{align} 
  where $ \overline{\be} $ lies between $ \widetilde{\be}_{S^{\prime}} $ and $ \be_{S^{\prime}}^{*} $, $ \mathcal{X}_{S^\prime}^{\top} = \left( (\boldsymbol{X}_{11})_{S^\prime}, \dots, (\boldsymbol{X}_{1n})_{S^\prime}, \ldots, (\boldsymbol{X}_{mn})_{S^\prime} \right)  $ is the sub-matrix of $ \mathcal{X}^{\top} $ indexed by model $ S^\prime $. 
  By the Courant-Fischer theorem, we know that the maximum singular value of sub-matrix $ \mathcal{X}_{S^\prime} $ is less than or equal to the maximum singular value of $ \mathcal{X} $. 
  Thus, we have $  \Vert \nabla \mathcal{L} ( \widetilde{\be}_{S^{\prime}}) - \nabla \mathcal{L}( \be_{S^{\prime}}^{*}) \Vert_2  \leq   (\rho_0 \mu/N)  \|  \widetilde{\be}_{S^{\prime}} -\be_{S^{\prime}}^{*}  \|_2 $, where $ \rho_0 $ is the largest eigenvalue of $ \mathcal{X}^{\top} \mathcal{X} $. 
  Similarly, it can be shown that $ \Vert \nabla \mathcal{L}_{1} ( \widetilde{\be}_{S^{\prime}} ) - \nabla \mathcal{L}_{1} \left( \be_{S^{\prime}}^{*} \right)  \Vert_2  \leq   (1/n)  \rho_1 \mu  \Vert  \widetilde{\be}_{S^{\prime}} -\be_{S^{\prime}}^{*}  \Vert_2 $, where $ \rho_1 $ is the largest eigenvalue of $ \mathbb{X}_1^{\top} \mathbb{X}_1 $. 
  
  Note that $ \lambda =O (q \sqrt{\ln p/n}) = O (N^{\tau_2} \sqrt{N^{\alpha}/n}) $, 
  which implies that $ \lambda \rightarrow 0 $ and $ N^{1 - \tau_1}  \geq n N^{-\tau_1} \geq \lambda q $ by Conditions (C1) and (C2) and $ n \gtrsim N^{(2/3) \tau_1 + 2\tau_2 + \alpha} $. 
By Condition (C1) and the fact that $ \Vert \widetilde{\be}_{S^{\prime}} -\be_{S^{\prime}}^{*} \Vert_2  \leq  \Vert  \widetilde{ \be } -\be^{*} \Vert_2 $, 
Lemma~\ref{lemma 2} implies 
  \begin{align} \label{ssp:eq3}
  \Pr\big\{ \Vert \nabla \mathcal{L} ( \widetilde{\be}_{S^{\prime}} ) - \nabla \mathcal{L}( \be_{S^{\prime}}^{*} ) \Vert_2 
  & \geq \left(c_{1} w_{1} / 6\right) N^{-\tau_{1}}  \big\}  \nonumber  \\
  &\leq   \Pr\big\{\Vert  \widetilde{\be}_{S^{\prime}} -\be_{S^{\prime}}^{*} \Vert_2 \geq (c_{1} w_{1} / (6 \rho_{0} \mu) ) N^{1-\tau_{1}}  \big\}    \nonumber  \\
  &\leq   2p \exp(- C_1 n \lambda^2 ).
  \end{align}
  Similarly, we can show that the second term on the right-hand side of \eqref{ssp:eq2} also satisfies 
  \begin{align} \label{ssp:eq4}
  \Pr\big\{ \Vert \nabla \mathcal{L}_{1} ( \widetilde{\be}_{S^{\prime}} ) - \nabla \mathcal{L}_{1} ( \be_{S^{\prime}}^{*}) \Vert_2 \geq \left(c_{1} w_{1} / 6\right) N^{-\tau_{1}} \big\} 
  \leq   2p \exp(- C_1 n \lambda^2 ).
  \end{align}

  For the third term on the right-hand side of~\eqref{ssp:eq2}, we have 
  \begin{equation} \label{ssp:eq5}
  \begin{aligned} 
    \Pr\{\Vert \nabla \mathcal{L} ( \be_{S^{\prime}}^{*}) \Vert_2 \geq (c_{1} w_{1} / 6) N^{-\tau_{1}}\}
    \leq  \sum_{l \in S^{\prime}} \Pr \big\{ (\nabla \mathcal{L}( \be_{S^{\prime}}^{*}))_l^{2}  \geq  (2k)^{-1} (c_{1} w_{1} / 6)^{2} N^{-2\tau_{1}}\big\}, 	 
  \end{aligned} 
  \end{equation}
  where $(\nabla \mathcal{L}( \be_{S^{\prime}}^{*}))_l$ denotes the $l$th component of $\nabla \mathcal{L}( \be_{S^{\prime}}^{*})$ and is defined as
  \[
    (\nabla \mathcal{L}( \be_{S^{\prime}}^{*}))_l
    =N^{-1}\sum\limits_{i=1}^{m} \sum\limits_{j=1}^{n}[  b^{\prime}(\boldsymbol{X}_{ijS^{\prime}}^{\top} \be_{S^{\prime}}^{*} ) - y_{ij}] x_{ij, l}
    =-N^{-1}\sum\limits_{i=1}^{m} \sum\limits_{j=1}^{n} [y_{ij}-\mathbb{E} (y_{ij}|\boldsymbol{X}_{ij})] x_{ij, l}.
  \]
  Let $ t_{Nij} = x_{ij,l} ( \sum_{i=1}^{m} \sum_{j=1}^{n} x_{ij,l}^{2} \sigma_{ij}^{2})^{-1/2}$,
  which satisfies $\sum_{i=1}^{m} \sum_{j=1}^{n} t_{Nij}^{2} \sigma_{ij}^{2}=1 $. 
Then, by some arguments similar to the proof of Lemma \ref{lemma 1}, we can show
  \begin{align*}
  \Pr \left\lbrace \left|  (\nabla \mathcal{L}( \be_{S^{\prime}}^{*}))_l \right|    
  \geq (2k)^{-1/2} (c_{1} w_{1} / 6) N^{-\tau_{1}} \right\rbrace  \leq 2\exp \left( -b_2 N^{1-2\tau_{1}-\tau_{2}} \right),
  \end{align*}
  where $ b_2$ is some positive constant.
This, together with \eqref{ssp:eq5} implies
  \begin{equation}  \label{ssp:third}
  \Pr\{\Vert \nabla \mathcal{L} ( \be_{S^{\prime}}^{*}) \Vert_2 \geq (c_{1} w_{1} / 6) N^{-\tau_{1}}\} \le 4 k \exp \left( -b_2 N^{1-2\tau_{1}-\tau_{2}} \right).  	
  \end{equation}  
Combining the results of~\eqref{ssp:eq2}, \eqref{ssp:eq3}, \eqref{ssp:eq4} and \eqref{ssp:third}, we obtain 
  \begin{align}  \label{ssp:eq9}
    \Pr \Big\{ \min_{S \in \mathcal{M}_{-}^{k}} \ell ( \be_{S^{\prime}} ) \leq \ell ( \be_{S^{\prime}}^{*} )\Big\} 
    \leq &\sum_{S \in \mathcal{M}_{-}^{k}} \Pr\big\{  \ell ( \be_{S^{\prime}} ) \leq \ell ( \be_{S^{\prime}}^{*} )\big\}   \nonumber  \\
    \leq & p^{k} 2p \exp(- C_1 n \lambda^2 )    
    +  p^{k} 4 k \exp \left( -b_2 N^{1-2\tau_{1}-\tau_{2}} \right)      \nonumber  \\
    \leq & b_3 \exp \left( b_4 N^{\tau_2 +\alpha} - b_5 N^{2 \tau_2 + \alpha}  \right) 
    +    b_6 \exp \big( \tau_{2} \ln N + b_7 N^{\tau_{2}+ \alpha} -b_8 N^{1-2 \tau_{1}-\tau_{2}} \big),
  \end{align}
  where $b_i~(i=3,\dots,8)$ are some positive constants and the last inequality holds by Conditions (C1) and (C2).
   If $ \tau_{1} + \tau_{2} < (1-\alpha)/2 $,
then the right-hand side of \eqref{ssp:eq9} tends to zero. 
Furthermore, since the variance $ b^{\prime \prime}(\cdot) \geq 0 $, we have that $\ell(\be_{S^\prime})$ is convex in $\be_{S^\prime}$.
  Therefore, these results still hold for any  $ \be_{S^{\prime}} $ that satisfies $ \left\| \be_{S^{\prime}} - \be_{S^{\prime}}^{*}\right\|_{2} \geq  w_{1} N^{-\tau_{1}} $.
  
  For any $ S \in \mathcal{M}_{-}^{k} $, 
  let $ \breve{\be}_{S^{\prime}} $ be $ \widehat{\be}_{S} $ augmented with zeros corresponding to the elements 
  in $ S^{\prime} \setminus S^{*} = S \setminus S^{*} $. 
  By Condition (C2),  we have $\| \breve{\be}_{S^{\prime}} - \be_{S^{\prime}}^{*}\|_{2} \geq \| \be_{S^{*} \setminus S}^{*}\|_{2} \geq w_{1} N^{-\tau_{1}} $. 
  Then by~\eqref{ssp:eq9}, we obtain
  \[  \Pr \Big\{ \min_{S \in \mathcal{M}_{-}^{k}} \ell(\hat{\be}_{S}) \leq \max_{S \in \mathcal{M}_{+}^{k}} \ell(\hat{\be}_{S})\Big\}  
  \leq   \Pr\Big\{\min_{S \in \mathcal{M}_{-}^{k}} \ell ( \breve{\be}_{S^{\prime}} ) \leq \ell (\be_{S^{\prime}}^{*} )\Big\} 
  = o(1).   \] 
  This completes the proof.
\end{proof}

\begin{proof}[Proof of Theorem~\ref{thm: DIHT screening property}]
  For simplicity, define $ d = \min_{j \in S^{*}}| \beta_{j}^{*}| $.
  To prove Theorem \ref{thm: DIHT screening property}, it suffices to show that for any $ t \geq 1$,  
  \[  \Pr\{ \Vert  \widehat{ \be }^{(t)}-\be^{*} \Vert _{\infty} < d/ 2 \} \rightarrow 1~~\text{as}~N\rightarrow \infty.  \] 
  We next show 
  \begin{equation} \label{algo:eq3}
  \| \widehat{ \be }^{(t)} - \be^{*} \|_{\infty} = o_{p}(d). 
  \end{equation}
  For this, we use the mathematical induction method.
  First, we prove that the statements hold when $ t=0 $. 
  Note that when $t=0$, $\widehat\be^{(0)}=\widetilde{ \be } $ is the Lasso estimator of $\be^*$ obtained by \eqref{initial:Lasso}.
  If Condition (C1) holds and $ \lambda n^{1/2}N^{-\alpha/2}  \to \infty $ as $N\rightarrow \infty$, then Lemma \ref{lemma 2} implies that
  $$ \Pr( \Vert  \widetilde{ \be } - \be^{*} \Vert_{1} \leq C_2 \lambda q ) \to 1 .$$
  In addition,  $ \lambda q = o(N^{-\tau_1})$ under Condition (C2) and $ \lambda = o ( N^{-(\tau_{1} + \tau_{2} )})$.
  These facts imply
  \[  \Vert  \widehat{ \be }^{(0)} - \be^{*} \Vert_{\infty} 
  =\Vert  \widetilde{ \be } - \be^{*} \Vert_{\infty}  
  \leq \Vert  \widetilde{ \be } - \be^{*} \Vert_{1} = o_{p}(d).  \]  
Thus, the statements in \eqref{algo:eq3} holds with $t=0$.
  
  We next show that if $\|\widehat\be^{(t-1)}-\be^*\|_{\infty}=o_p(d)$,
  then $\|\widehat\be^{(t)}-\be^*\|_{\infty}$ is also of order $o_p(d)$.
  Note that $ \widehat{ \be }^{(t)} = \mathbb{H}(\widehat{ \boldsymbol{\gamma} }^{(t-1)},r_k)$,
  where $\widehat{ \boldsymbol{\gamma} }^{(t-1)} = \widehat{ \be }^{(t-1)} -\vartheta^{-1} \nabla \ell (\widehat{ \be }^{(t-1)})$, $r_k$ is the $k$-th largest component of $|\widehat{\boldsymbol{\gamma} }^{(t-1)}|$, 
  and $\mathbb{H}(x,a)=x\mathbb{I}(|x|\ge a)$ denotes the hard-thresholding function.
  If $ \Vert  \widehat{ \boldsymbol{\gamma} }^{(t-1)} - \be^{*} \Vert _{\infty} = o_{p}(d) $, 
  then we have $ \Vert  \widehat{ \boldsymbol{\gamma} }^{(t-1)}_{S^*_c} \Vert _{\infty} = o_{p}(d) $ 
  and $ \Vert \widehat{ \boldsymbol{\gamma} }^{(t-1)}_{S^*} \Vert _{\infty} = O_{p}(d)$. 
  This implies that with probability approaching one, 
  the components in $ \widehat{ \boldsymbol{\gamma} }^{(t-1)}_{S^*} $ are among the ones with the largest absolute values among all components.
  Therefore, if $ \Vert  \widehat{ \boldsymbol{\gamma} }^{(t-1)} - \be^{*} \Vert _{\infty} = o_{p}(d) $, then we have  
  \begin{equation} \label{algo:eq4}
  \| \widehat{ \be }^{(t)} - \be^{*}\|_{\infty} 
  = \|\mathbb{H}(\widehat{ \boldsymbol{\gamma} }^{(t-1)},r_k) - \be^{*}\|_{\infty}    \leq    \|\widehat{ \boldsymbol{\gamma} }^{(t-1)} - \be^{*} \|_{\infty} = o_{p}(d). 
  \end{equation}
  Now, we need to prove $ \Vert  \widehat{ \boldsymbol{\gamma} }^{(t-1)} - \be^{*} \Vert _{\infty} = o_{p}(d) $,
  that is,
  \[  \Vert  \widehat{ \be }^{(t-1)} -\vartheta^{-1} \nabla \ell (\widehat{ \be }^{(t-1)}) - \be^{*}  \Vert _{\infty} = o_{p}(d).  \] 
Since 
  \begin{align*}
      \| \widehat{ \be }^{(t-1)} -\vartheta^{-1} \nabla \ell ( \widehat{ \be }^{(t-1)}) - \be^{*}\|_{\infty}  
  \leq  \| \widehat{ \be }^{(t-1)} - \be^{*}\|_{\infty} +\|\vartheta^{-1} \nabla \ell (\widehat{ \be }^{(t-1)})\|_{\infty}, 
  \end{align*}
and $ \| \widehat{ \be }^{(t-1)} - \be^{*}\|_{\infty}=o_p(d)$ under the induction assumption,
it suffices to show that $\|\vartheta^{-1} \nabla \ell (\widehat{ \be }^{(t-1)})\|_{\infty}=o_p(d)$.
  We begin with the following decomposition
  \begin{align}  \label{algo:eq5}
    &\|\nabla \ell (\widehat{ \be }^{(t-1)})\|_{\infty}
    =     \| \nabla \mathcal{L}_{1}( \widehat{ \be }^{(t-1)}) - \nabla \mathcal{L}_{1}(\widetilde{ \be }) + \nabla \mathcal{L}(\widetilde{ \be })\|_{\infty}  \nonumber   \\
    \leq & \| \nabla \mathcal{L}_{1}( \widehat{ \be }^{(t-1)}) - \nabla \mathcal{L}_{1}(\be^*)\Vert_{\infty}
    +  \Vert \nabla \mathcal{L}_{1}(\widetilde{ \be }) - \nabla \mathcal{L}_{1}(\be^*) \Vert_{\infty}    \nonumber      \\
     & +    \Vert \nabla \mathcal{L}(\widetilde\be) - \nabla \mathcal{L}(\be^*) \Vert_{\infty}
    + \Vert  \nabla \mathcal{L}(\be^*) \Vert_{\infty}.
  \end{align}
  Note that $ \nabla^2 \mathcal{L}_{1}({\be}) = (1/n) \mathbb{X}_1^{\top} \text{diag}(b^{\prime\prime} ( \boldsymbol{X}_{11}^{\top} {\be}), \ldots, b^{\prime\prime} ( \boldsymbol{X}_{1n}^{\top} {\be}))  \mathbb{X}_1 $ and $b^{\prime\prime} (\theta)\le \mu $.
  Additionally, $|x_{ij,l}|\le c_2$ under Condition (C4) 
  and the number of non-zero entries in $(\widehat{ \be }^{(t-1)} - \be^{*})$ is at most $2k$.
  By the Taylor's expansion, the first term on the right-hand side of \eqref{algo:eq5} satisfies
  \begin{align}\label{algo:eq6}
    \| \nabla \mathcal{L}_{1}( \widehat{ \be }^{(t-1)}) - \nabla \mathcal{L}_{1}(\be^*)\Vert_{\infty}
    =     \Vert  \nabla^2 \mathcal{L}_{1}(\overline{\be}_1) ( \widehat{ \be }^{(t-1)} - \be^{*}) \Vert_{\infty} 
    \le    2\mu c_2^2 k\Vert  \widehat{ \be }^{(t-1)} - \be^{*} \Vert_{\infty},
  \end{align}
  where $ \overline{\be}_1 $ is between $\widehat\be^{(t-1)}$ and $\be^*$.
  This implies that 
  \begin{align*}
  \vartheta^{-1}\|\nabla \mathcal{L}_{1}( \widehat{ \be }^{(t-1)}) - \nabla \mathcal{L}_{1}(\be^*)\Vert_{\infty}=o_p(kd/\vartheta),
  \end{align*}
  which is of order $ o_p(d) $ with $ \vartheta\ge k $.
  Similarly, the second and third terms on the right-hand side of~\eqref{algo:eq5} satisfy
  \begin{align*}
    &\Vert \nabla \mathcal{L}_{1}(\widetilde{ \be }) - \nabla \mathcal{L}_{1}(\be^*) \Vert_{\infty}
    =  \Vert  \nabla^2 \mathcal{L}_{1}(\overline{\be}_2) (\widetilde{ \be } - \be^{*}) \Vert_{\infty} =  O(k) \|\widetilde{ \be } - \be^{*}\|_{\infty}       \\
    \text{and}~~& \Vert\nabla \mathcal{L}(\widetilde\be) - \nabla \mathcal{L}(\be^*)\Vert_{\infty}
    =  \Vert  \nabla^2 \mathcal{L}(\overline{\be}_3) (\widetilde{ \be } - \be^{*}) \Vert_{\infty}  =  O(k) \|\widetilde{ \be } - \be^{*}\|_{\infty},
  \end{align*}
  where $\overline{\be}_2$ and $\overline{\be}_3$ are between $\widehat\be^{(t-1)}$ and $\be^*$. 
  Since $ \|\widetilde{ \be } - \be^{*}\|_{\infty} = o_p(d) $ and $\vartheta\ge k$,  
  we have $ \vartheta^{-1} \Vert \nabla \mathcal{L}_{1}(\widetilde{ \be }) - \nabla \mathcal{L}_{1}(\be^*) \Vert_{\infty} = o_p(d) $  and $ \vartheta^{-1} \Vert\nabla \mathcal{L}(\widetilde\be) - \nabla \mathcal{L}(\be^*)\Vert_{\infty} = o_p(d) $.
  
By Condition (C2) and $ \vartheta N^{- \tau_1} \to 0 $, we have $ \vartheta d \to 0 $ as $N\rightarrow \infty.$
In addition,  by some arguments similar to the proof of Lemma \ref{lemma 1},
we have
  \begin{align*}
    \Pr\big(\vartheta^{-1} \Vert  \nabla \mathcal{L}(\be^*) \Vert_{\infty} > b_9 \vartheta d^2 \big)
     \leq  \sum_{l=1}^{p} \Pr \left( | (\nabla \mathcal{L}(\be^*))_l | >b_9 (\vartheta d)^2 \right)  
     \leq 2\exp(b_{10}N^{\alpha} - b_{11}\vartheta^4 N^{1-4\tau_1}),
  \end{align*}
  where $b_9, b_{10}$ and $b_{11}$ are some positive constants.
  Therefore, if $\vartheta N^{(1-\alpha-4\tau_1)/4} \to \infty $ as $N\rightarrow\infty$, 
  then the last term on the right-hand side of \eqref{algo:eq5} satisfies $\vartheta^{-1} \Vert  \nabla \mathcal{L}(\be^*) \Vert_{\infty} = O_p(\vartheta d^2)=o_p(d) $.
  These facts imply that $\|\vartheta^{-1} \nabla \ell (\widehat{ \be }^{(t-1)})\|_{\infty}=o_p(d)$.
  Therefore, \eqref{algo:eq4} holds. This completes the proof.
\end{proof}

\end{document}